\PassOptionsToPackage{table}{xcolor}
\documentclass{article} 
\usepackage{iclr2027_conference,times}
\IfFileExists{latexml.sty}{\usepackage{latexml}}{\newif\iflatexml\latexmlfalse}

\usepackage{graphicx}
\graphicspath{{./Figs/}{./}}

\usepackage{amsmath,amsfonts,bm}

\def\eqref#1{equation~\ref{#1}}

\def\1{\bm{1}}

\DeclareMathAlphabet{\mathsfit}{\encodingdefault}{\sfdefault}{m}{sl}
\SetMathAlphabet{\mathsfit}{bold}{\encodingdefault}{\sfdefault}{bx}{n}

\usepackage{wrapfig}
\usepackage{amsmath,amssymb}
\usepackage{amsthm}
\usepackage{url}
\usepackage{enumitem}
\usepackage{fontawesome5}
\usepackage{dsfont}
\usepackage[table]{xcolor}
\usepackage{booktabs}
\usepackage{multirow}
\usepackage{algorithm}
\DeclareRobustCommand{\bestmark}[1]{%
  {\setlength{\fboxsep}{1pt}\colorbox{bestcolor}{\textbf{#1}}}}
\DeclareRobustCommand{\secondmark}[1]{%
  {\setlength{\fboxsep}{1pt}\colorbox{secondcolor}{\underline{#1}}}}
\usepackage{algpseudocode}   

\definecolor{bestcolor}{RGB}{204,229,255}
\definecolor{secondcolor}{RGB}{255,230,204}

\definecolor{llmblue}{HTML}{1F5FA8}
\definecolor{physred}{HTML}{B03A2E}
\newcommand{\phase}[1]{\Statex\hspace{-\algorithmicindent}\textcolor{gray!85}{\textit{// #1}}}
\newcommand{\llm}[1]{\textcolor{llmblue}{#1}}
\newcommand{\Eval}[1]{\textcolor{physred}{\textsc{eval}$(#1)$;\ $\mathrm{NFE}\mathrel{+{=}}1$}}

\newtheorem{lemma}{Lemma}
\newtheorem{assumption}{Assumption}

\usepackage[most]{tcolorbox}
\definecolor{absbg}{HTML}{EEF3F8}      
\definecolor{corrblue}{HTML}{1F6FD0}   
\iflatexml
\newenvironment{abstractbox}{\begin{abstract}}{\end{abstract}}
\else
\newtcolorbox{abstractbox}{
  enhanced,
  colback=absbg,
  colframe=absbg,
  boxrule=0pt,
  arc=3mm,
  left=5mm, right=5mm, top=4mm, bottom=4mm,
  before skip=-4pt, after skip=12pt
}
\fi

\fancypagestyle{firstpage}{%
  \fancyhf{}%
  \lhead{\raisebox{-0.75cm}[0pt][0pt]{%
    \includegraphics[height=1.3cm]{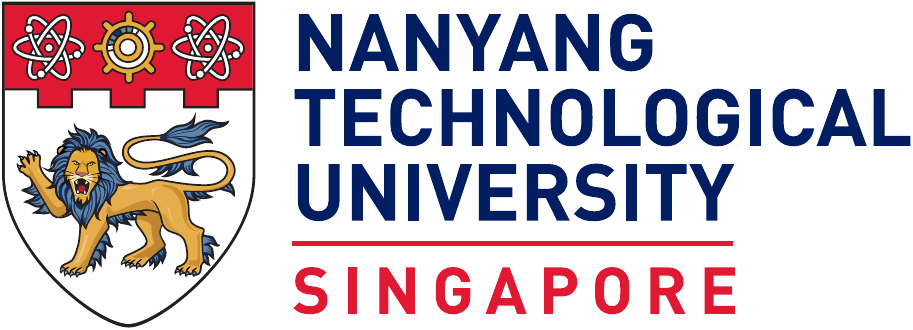}%
    \hspace{0.45cm}{\color{black!25}\rule[0.13cm]{0.5pt}{1.04cm}}\hspace{0.45cm}%
    \includegraphics[height=1.3cm]{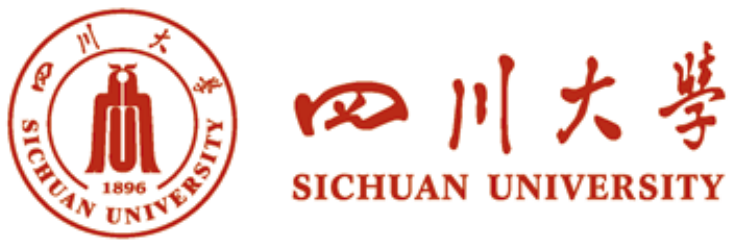}}}%
  \cfoot{\thepage}%
  \renewcommand{\headrulewidth}{0pt}%
}

\usepackage{hyperref}
\hypersetup{
    colorlinks=true,
    citecolor=red!70!black,
    linkcolor=blue!60!black,
    urlcolor=green!50!black,
    pdftitle={Vibe Building},
}

\title{Vibe Building}

\iflatexml
\author{
Yongqing Jiang$^{\dagger}$ \\ Nanyang Technological University, Singapore
\and Haoran Luo$^{*}$ \\ Nanyang Technological University, Singapore
\and Jianze Wang \\ Sichuan University, China
\and Xin Zhou \\ Nanyang Technological University, Singapore
\and Kaoshan Dai \\ Sichuan University, China
\and Zhiqi Shen$^{*}$ \\ Nanyang Technological University, Singapore
}
\else
\author{
Yongqing Jiang$^{1\dagger}$, Haoran Luo$^{1*}$, Jianze Wang$^{2}$, Xin Zhou$^{1}$, Kaoshan Dai$^{2}$ \& Zhiqi Shen$^{1*}$ \\
$^{1}$Nanyang Technological University, Singapore \quad $^{2}$Sichuan University, China
}
\fi

\iclrfinalcopy 
\begin{document}

{\setlength{\tabcolsep}{0pt}\maketitle}   
\thispagestyle{firstpage}                 
\lhead{\textsc{Vibe Building}}            
\iflatexml
\begin{center}\small
\textsuperscript{\textdagger}Project leader.\quad \textsuperscript{*}Corresponding authors.\quad
Correspondence to \href{mailto:yong.qingjiang@outlook.com}{yong.qingjiang@outlook.com}.
\end{center}
\else
{\renewcommand{\thefootnote}{}%
 \footnotetext{$^{\dagger}$Project leader.\quad $^{*}$Corresponding authors.\quad
 Correspondence to \href{mailto:yong.qingjiang@outlook.com}{\textcolor{black!80}{\texttt{yong.qingjiang@outlook.com}}}.}}
\fi

\begin{abstractbox}
Automated building design must comply with seismic and wind codes and satisfy structural mechanics constraints, yet most existing agents produce visually plausible models without verification grounded in mechanical analysis and code compliance. We introduce the Vibe Building task and propose \textbf{PE-Loop} (\textbf{P}hysics-\textbf{E}ngine-in-the-\textbf{Loop}), an agent in which a deterministic physics engine is the sole source of evaluation signals, mapping code constraints to a physics process reward, while the language model is confined to proposing discrete revisions (section menu, topology, and lateral system). Designs are verified by held-out seismic and wind time-history checks and a constructability gate. On VB-Bench, 3,577 physics-adjudicated building instances across six code families, PE-Loop achieves the highest verified success rate under three of four backbone LLMs, the highest held-out seismic pass rate under all four, and the highest held-out wind pass rate under three. Replacing the physics verdict with a language-model judge, all else fixed, leaves 58.43\% of accepted designs noncompliant. These results suggest that reliable structural design rests less on a stronger LLM proposer than on an adjudicator the proposer cannot influence, a division of labor for agents whose outputs must hold up in the physical world.

\vspace{0.8em}
{\sffamily\small
\begin{tabular}{@{}l@{\quad}l@{}}
\makebox[1.3em][l]{\color{corrblue}\faGlobe}\textbf{\color{corrblue}Project Page:} & \href{https://jovanqing.github.io/Vibe-Building/}{\textcolor{black!80}{https://jovanqing.github.io/Vibe-Building}} \\
\makebox[1.3em][l]{\color{corrblue}\faGithub}\textbf{\color{corrblue}Code:}         & \href{https://github.com/Jovanqing/Vibe-Building}{\textcolor{black!80}{https://github.com/Jovanqing/Vibe-Building}}
\end{tabular}}
\end{abstractbox}
\section{Introduction}

\begin{wrapfigure}{r}{0.5\textwidth}
\vspace{-10pt}
\centering
\includegraphics[width=0.48\textwidth]{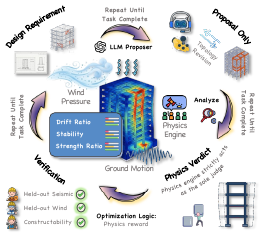}
\caption{An illustration of PE-Loop.}
\label{fig:abstract}
\vspace{-10pt}
\end{wrapfigure}

In vibe coding \citep{karpathy2025vibe}, a developer states intent in natural language and accepts whatever the model produces, as long as it runs. The criterion is permissive because errors are cheap: a faulty program can be reverted and regenerated within seconds. Buildings afford no such recourse. A structural design must satisfy seismic and wind codes and preserve a safety margin that dead loads erode over decades \citep{meng2026physics}; an oversight at the modeling stage compounds silently into irreversible failure \citep{jiang2026rethinking,chen2026post}. Automating design has long been a goal of digital engineering \citep{makoond2024arresting}. We ask what vibe coding becomes when ``runs'' is replaced by ``stands,'' and argue that the answer depends not on how a design is generated but on \emph{who is allowed to accept it}.

\begin{figure}[!t]
\centering
\includegraphics[width=0.90\textwidth]{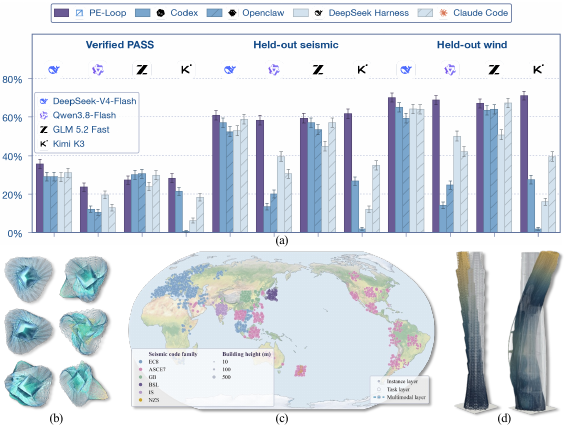}
\caption{\textbf{Swap the backbone, keep the reliability.}
\textbf{(a)}~Across four LLMs, the held-out pass rate of PE-Loop varies by 3.3\,pp (seismic) and 4.1\,pp (wind), versus 23.8 and 27.7\,pp for the best baseline per backbone; whiskers: $\pm1$ SE.
\textbf{(b)}~First six plan-view modes of a Shanghai Tower model; torsional mode~5 appears only in plan.
\textbf{(c)}~Instance locations by code family. \textbf{(d)}~Two VB-Bench instances.}
\label{fig:teaser}
\end{figure}

LLM agents now write and run finite-element code \citep{qi2025feagpt,geng2025lightweight} and automate code checking \citep{liang2025automating,madireddy2025bim} (Fig.~\ref{fig:abstract}), but the physics engine is usually passive: it executes a committed solution or verifies a known answer. Generative layout methods \citep{liao2021structgan,gu2024diffusion,radler2025diverse} let physics guide the search but cannot read a natural-language brief. Loop- and graph-engineering paradigms \citep{wang2026looparena,chaoyue2026awesome_loop_engineering,feng2026graph,he2026awesomegraphengineering} supply the orchestration for such multi-stage tasks, yet a language model still commonly renders the verdict, which is neither reliable nor auditable \citep{zheng2023judging,gu2024llmasjudge,boiko2023autonomous}, and without external feedback such a model does not reliably correct its own errors \citep{huang2024cannot,kambhampati2024llmmodulo}. Unit tests \citep{jimenez2024swebench} and verifier-driven search \citep{novikov2025alphaevolve,guo2025deepseekr1} move the verdict to a program, but against a few targets or one scalar score, not the coupled constraints a structure must satisfy at once. We therefore hold that only an unpersuadable arbiter should accept a design: a deterministic physics engine that evaluates each candidate under the code-prescribed wind \citep{gb50009} and seismic \citep{gb50011} envelope and returns a per-constraint margin that fluent but unfaithful text cannot alter.

We instantiate this position as \textbf{PE-Loop} (\textbf{P}hysics-\textbf{E}ngine-in-the-\textbf{Loop}), a generate-verify-reflect loop that terminates only when every strength and serviceability constraint holds. The physics engine holds sole acceptance authority, and its margins form a physics process reward (PPR) that couples one-dimensional section search with topology repair; the language model only proposes discrete revisions. Because the verdict never depends on the proposer, the held-out pass rate of PE-Loop is nearly invariant across four backbones, whereas the best baseline varies by more than 20\,pp (Fig.~\ref{fig:teaser}a). To evaluate it, we formalize Vibe Building, which keeps the vibe-coding interaction but moves the verdict to a physics engine, and build VB-Bench: 3,577 physics-adjudicated instances across six code families, with held-out dynamic re-checks and a constructability gate. Concurrently, \citet{luo2026verification} close a parameter-level repair loop around a finite-element verifier on 44 member-scale cases; we add topology-level repair, which they identify as the limit of such loops, held-out dynamic loading, and a counterfactual that swaps only the verdict holder, on a benchmark two orders of magnitude larger (Appendix~\ref{sec:related}).

\begin{figure}[!t]
\centering
\includegraphics[width=0.95\textwidth]{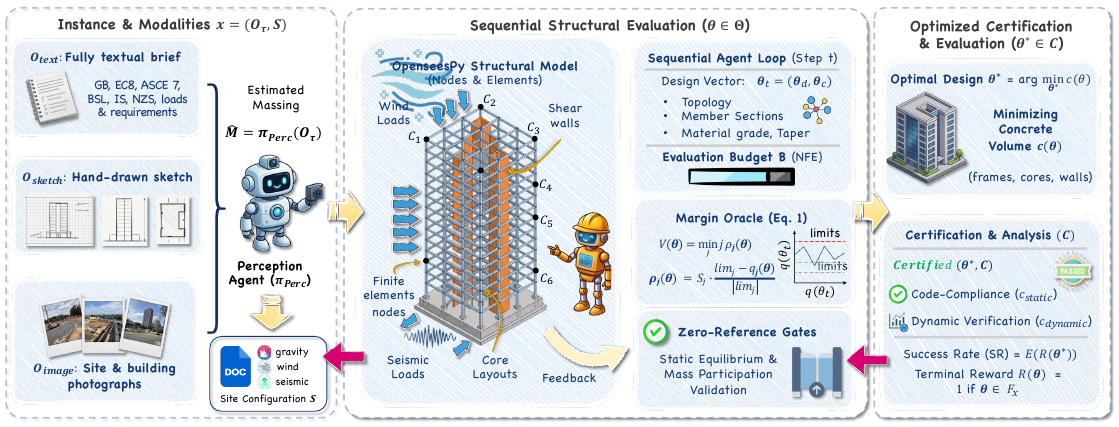}
\caption{The definition and overview of the Vibe Building task.}
\label{fig:task}
\end{figure}

Taken together, our work makes the following contributions:
\begin{itemize}[leftmargin=*,itemsep=2pt,topsep=2pt]
\item \textbf{Vibe Building.} We formalize the generation of code-compliant building designs from multimodal input, with a protocol combining in-loop compliance, held-out dynamic re-checks, and a constructability gate.
\item \textbf{VB-Bench.} We construct 3,577 physics-adjudicated instances across six code families, 1,311 gradable tasks, and a structural-code knowledge base in which every limit is traceable to its clause.
\item \textbf{PE-Loop.} We propose an agent in which a deterministic physics engine holds sole acceptance authority and its constraint-wise margins drive both section search and topology repair.
\item \textbf{Physics adjudication is decisive.} PE-Loop achieves the best verified success under three of four backbones and the best held-out seismic pass rate under all four; replacing its verdict with an LLM judge turns 58.43\% of accepted designs into false positives.
\end{itemize}

\section{Task: Vibe Building}
\label{sec:task}

\paragraph{Instances and modalities.}
An instance $x = (O_\tau, S)$ pairs an \emph{observation} $O_\tau$ of an underlying massing model $M$ with a \emph{site} $S$ that fixes the code family (GB, EC8, ASCE~7, BSL, IS, or NZS; \citealp{gb50011,en1998,asce7,bsl_japan,is1893,nzs11705}) and its gravity, wind, and seismic loading provisions. The observation takes one of three modalities, $\tau \in \{\mathrm{text}, \mathrm{sketch}, \mathrm{image}\}$: a numeric text brief, a sketch with text, or photographs alone (Fig.~\ref{fig:task}).

\paragraph{Constraints and physical margins.}
Each instance carries a constraint set $G_x = \{g_1, \dots, g_m\}$, where each $g_j$ bounds a demand quantity $q_j(\theta)$ of a candidate realization $\theta \in \Theta$ against a code-prescribed limit $\mathrm{lim}_j$; the space $\Theta=\Theta_d\times\Theta_b$ factorizes into discrete choices (lateral system, core and wall counts, core plan geometry, and a section triple of beam scale, concrete grade, and taper on a finite menu) and the square-column side $b$ on a 50-mm grid. Demand quantities are evaluated under the code-prescribed envelope of gravity, wind, and seismic load cases. Two zero-reference gates (static equilibrium; effective mass participation) reject invalid models before any code limit is consulted. Each constraint contributes a signed, normalized margin
\begin{equation}
\rho_j(\theta) = s_j \, \frac{\mathrm{lim}_j - q_j(\theta)}{|\mathrm{lim}_j|},
\qquad
V(\theta) = \min_j \rho_j(\theta),
\label{eq:margin}
\end{equation}
where $s_j \in \{\pm 1\}$ encodes the inequality direction; $\rho_j(\theta) \ge 0$ iff $g_j$ is satisfied, and $V(\theta)$ is the worst-case margin. We write $\boldsymbol{\rho}(\theta) = (\rho_1(\theta), \dots, \rho_m(\theta))$ for the full margin vector and $\boldsymbol{q}(\theta) = (q_1(\theta), \dots, q_m(\theta))$ for the demand vector, both returned to the agent at every evaluation. Each instance additionally carries a detailing set $G^{\mathrm{build}}_x$ of code clauses that bound member proportions of a realization, such as the beam clear-span-to-depth ratio; their margins follow Eq.~\ref{eq:margin} with $q_j$ a geometric ratio of $\theta$ and cost no simulator call.

\begin{figure}[!t]
\centering
\includegraphics[width=\textwidth]{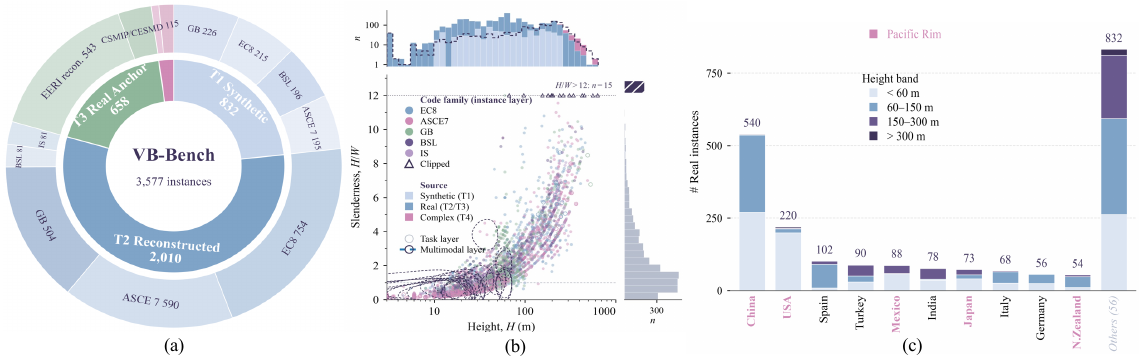}
\caption{\textbf{VB-Bench overview.} (a) Instance composition by route (inner) and code family (outer). (b) Height--slenderness: instance (filled), task (open circles), multimodal (dashed). (c) Geographic coverage (top-10 countries, Pacific Rim in pink).}
\label{fig:bench}
\end{figure}

\paragraph{Episode protocol and objective.}
An episode on instance $x$ unfolds under the margin oracle of Eq.~\ref{eq:margin}. It is a sequential interaction bounded by an evaluation budget $B$ counted in simulator calls (number of function evaluations, NFE): at step $t$ the agent selects $\theta_t \in \Theta$ and observes $\big(\boldsymbol{q}(\theta_t), \boldsymbol{\rho}(\theta_t)\big)$; the episode ends when the agent returns a final design $\theta^\star$, when the budget is exhausted before a design is returned, or when the agent declares \textsc{escalated}, a certifiable admission that the available design levers are insufficient, which is itself a scored outcome. Writing $F_x = \{\theta \in \Theta : \bigwedge_{j=1}^{m} g_j(\theta)\}$ for the feasible set, the terminal reward is sparse, $R(\theta) = \mathds{1}[\theta \in F_x]$, and the ideal return minimizes material among feasible realizations:
\begin{equation}
\theta^\star = \operatorname*{arg\,min}_{\theta \in F_x} c(\theta),
\label{eq:objective}
\end{equation}
with $c(\theta)$ the concrete volume of the realized structure, accounting for frames, cores, and walls alike.

\paragraph{Certification and multimodal grounding.}
The return of Eq.~\ref{eq:objective} is certified before it counts: the agent returns $(\theta^\star, \mathcal{C})$, where the certificate $\mathcal{C}=(\boldsymbol{\rho},\boldsymbol{\rho}^{\mathrm{L4}},\boldsymbol{\rho}^{\mathrm{build}})$ records, respectively, the code-compliance margins, a held-out dynamic verification under recorded ground motions and field-recorded or spectrum-consistent wind series that the agent never optimizes against, and compliance with the detailing set $G^{\mathrm{build}}_x$. For $\tau \in \{\mathrm{sketch}, \mathrm{image}\}$, the agent first applies a perception operator $\pi_{\mathrm{perc}}$ to its observation, yielding a massing estimate $\hat{M} = \pi_{\mathrm{perc}}(O_\tau)$, but adjudication is always performed on the ground-truth $M$, so design errors are never confounded with perception errors. The primary metric aggregates these certified returns over an evaluation set $\mathcal{X}$ of instances into the one-shot success rate
\begin{equation}
\mathrm{SR} = \mathbb{E}_{x \sim \mathcal{X}}\big[ R^{\mathrm{cert}}(\theta^\star_x) \big],
\label{eq:sr}
\end{equation}
where $\theta^\star_x$ is the design returned on instance $x$ and $R^{\mathrm{cert}}$, defined in Eq.~\ref{eq:m-l4}, is the conjunction of the three certificate components.

\section{Benchmark: VB-Bench}
\label{sec:benchmark}
VB-Bench has two layers (Fig.~\ref{fig:bench}; details in Appendix~\ref{app:bench}). The instance layer includes 3,577 physics-adjudicated buildings from four construction routes, forming a difficulty gradient: synthetic (T1, 832), literature-grade reconstructions (T2, 2,010), real anchors (T3, 658), and complex-massing towers (T4, 77). The task layer includes 1,311 gradable tasks stratified by solvability and question type (intent L1, design L2, modeling L3), with 1,197 paired variants and a 372-task quarantine set against contamination. Every sample passes expert review, physical adjudication under envelope gravity, wind, and seismic load cases with height-tiered limits, and time-history verification. Every limit, including that from detailing clause $l_n/h_b\ge 4$ of $G^{\mathrm{build}}_x$, carries a clause number traceable to a structural-code knowledge base compiled from 86 codes across six code families.

\begin{figure}[!t]
\centering
\includegraphics[width=\textwidth]{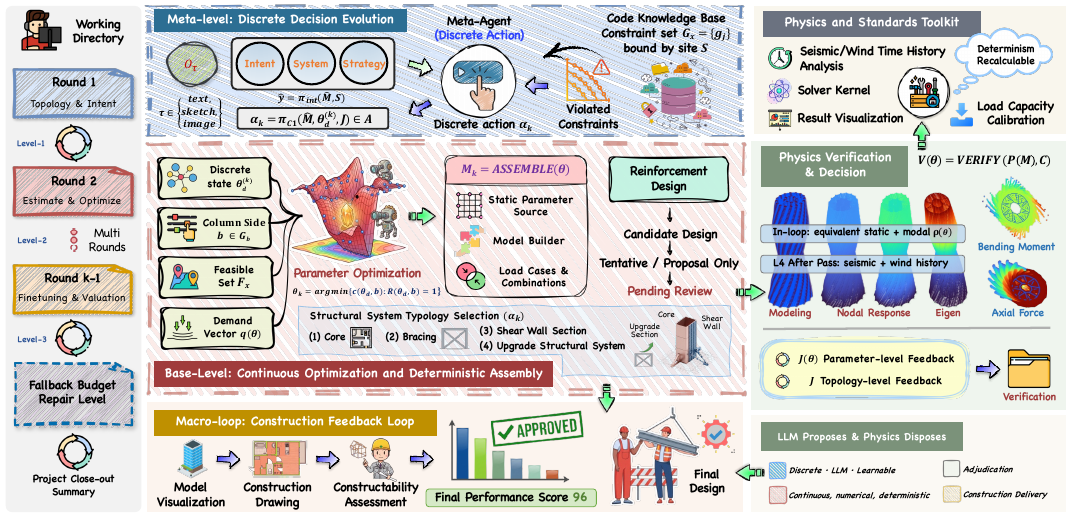}
\caption{Overview of PE-Loop.}
\label{fig:method}
\end{figure}

\section{Method: Physics-Engine-in-the-Loop}
\label{sec:method}

PE-Loop rests on one invariant: the only authority that advances a design is the deterministic margin oracle of Section~\ref{sec:task}, never the self-assessment of a language model~\citep{zheng2023judging,shinn2023reflexion}; agent outputs change the next candidate $\theta_t$, never the verdict. Constraints analytic in $\theta$ are enforced by construction, and only those needing the simulator are adjudicated. Figure~\ref{fig:method} gives the overview; pseudocode (Algorithm~\ref{alg:peloop}), proofs, and implementation details are in Appendix~\ref{app:method}.

\paragraph{Physics process signal.}
Because $\rho_j(\theta)\ge 0$ iff $g_j$ holds (Eq.~\ref{eq:margin}), feasibility and reward are functions of the margins alone, and the violated set directs repair:
\begin{equation}
R(\theta)=\mathds{1}[V(\theta)\ge 0]=\prod_{j=1}^{m}\mathds{1}[\rho_j(\theta)\ge 0],
\label{eq:m-consistency}
\end{equation}
\begin{equation}
\mathcal{J}(\theta)=\{\,j:\ \rho_j(\theta)<0\,\}.
\label{eq:m-violated}
\end{equation}
We call $\boldsymbol{\rho}$ the \emph{physics process reward} (PPR) and $R$ the verifiable outcome reward~\citep{uesato2022process,lightman2024verify,shao2024deepseekmath}; $R(\theta_t)$ is observable at every step. The loop consumes only the sign of $\boldsymbol{\rho}$; magnitudes enter only the certificate (Eq.~\ref{eq:m-l4}).

\paragraph{Design space and intent.}
The agent searches one coordinate on a grid, the square-column side $b$, and treats everything else as discrete:
\begin{equation}
\theta=(\theta_d,b),\quad
\theta_d=(n_{\mathrm{core}},n_{\mathrm{wall}},\kappa,\varsigma),\quad
b\in G_b=\{0.30+0.05\,i\ \text{m}:\ i=0,\dots,44\},
\label{eq:m-space}
\end{equation}
with core and wall counts $n_{\mathrm{core}},n_{\mathrm{wall}}$, core plan geometry $\kappa$, and section triple $\varsigma$ (beam scale, concrete grade, taper) on a finite menu (Appendix~\ref{app:intent}). For sketch and image inputs one vision-language call gives $\hat M=\pi_{\mathrm{perc}}(O_\tau)$; an intent call $\pi_{\mathrm{int}}(\hat M,S)$ then fixes the initial core and wall counts (Eq.~\ref{eq:m-intent}, Appendix~\ref{app:intent}). Neither call is charged to $B$, and adjudication stays on the true $M$.

\paragraph{Column-axis search (C2).}
For fixed $\theta_d$ write $\tilde\theta(b)=\Pi_{\mathrm{build}}(\theta_d,b)$ for the projected candidate. We adopt Assumption~\ref{ass:mono} (monotone margin along the projected column axis): $b\le b'\Rightarrow V(\tilde\theta(b))\le V(\tilde\theta(b'))$. The feasible sides then form an up-set of $G_b$, and since concrete volume is strictly increasing in $b$ (Lemmas~\ref{lem:bisect}--\ref{lem:cost}, Appendix~\ref{app:c2}),
\begin{equation}
b^\star(\theta_d)=\min\{b\in G_b:\ R(\theta_d,b)=1\}=\operatorname*{arg\,min}\{c(\theta_d,b):\ b\in G_b,\ R(\theta_d,b)=1\},
\label{eq:m-bstar}
\end{equation}
which exists iff $R(\theta_d,b_{\max})=1$. C2 finds it by an upper-bound probe, $\lceil\log_2 45\rceil=6$ bisection steps, and a confirmation, which together take at most $8$ evaluations; with one re-verification of the returned design, one round costs at most $9$. Throughout C2, $\varsigma$ stays frozen, up to the beam-scale clip of $\Pi_{\mathrm{build}}$ in the projected variant.

\paragraph{Discrete decision (C1) and outer loop.}
Round $k$ runs C2 and returns \textsc{pass} if $b^\star$ exists, or \textsc{escalated} if $k=K$ or $\mathrm{NFE}\ge B$. Otherwise, the language model picks one action from $\mathcal{A}$ at the section, topology, or system level (adjust $\varsigma$; add a core or wall; adjust $\kappa$; change the lateral system):
\begin{equation}
\theta_d^{(k+1)}=a_k\big(\theta_d^{(k)}\big),\qquad
a_k=\pi_{\mathrm{C1}}\big(\hat M,\ \theta_d^{(k)},\ \mathcal{J}_k\big)\in\mathcal{A},
\label{eq:m-outer}
\end{equation}
where $\mathcal{J}_k$ is the violated set of the last evaluation of round $k$: that of $(\theta_d^{(k)},b_{\max})$ when the probe fails, or of the failed confirmation when the budget truncates the bisection; the split follows comparative advantage: grid search beats the model on the column axis at equal budget, while the model supplies cross-instance experience for discrete repairs. We use $K=3$ (one decision per lever level) and $B=15$, shared across all compared systems. In the \texttt{ultra} tier, a material loop C3 further lowers concrete volume along $\varsigma$, accepting only feasible, strictly cheaper candidates (Eq.~\ref{eq:m-c3}, Appendix~\ref{app:c3}); the four tiers form one anytime family (Appendix~\ref{app:effort}).

\paragraph{Constructability by construction.}
The detailing clause $l_n/h_b\ge 4$ is analytic in $\theta$, so before every evaluation, a projection $\Pi_{\mathrm{build}}$ clips the beam-depth scale and caps $b_{\max}$ such that a compliant beam always exists, at no oracle cost (Eq.~\ref{eq:m-proj}, Appendix~\ref{app:build}). Assumption~\ref{ass:mono} is stated on the projected path, so Lemma~\ref{lem:bisect} applies to the capped grid verbatim, and Lemma~\ref{lem:safe} guarantees feasibility of every returned design even where the assumption fails. We evaluate $\Pi_{\mathrm{build}}$ as an ablation under one backbone (Section~\ref{sec:ablations}); the ``PE-Loop'' rows of Table~\ref{tab:m1} run without it, so that their Build column measures the detailing gap that search alone leaves open.

\paragraph{Final verification (L4) and certificate.}
The returned design undergoes one seismic response-history analysis under recorded ground motions and one wind time-history analysis, both elastic single-direction analyses (Appendix~\ref{app:method}), neither of which ever enters the loop, yielding held-out margins $\boldsymbol{\rho}^{\mathrm{L4}}$ (drift ratio; peak acceleration above 150~m). With the detailing margins $\boldsymbol{\rho}^{\mathrm{build}}$, the certificate and certified success are
\begin{equation}
\mathcal{C}=\big(\boldsymbol{\rho},\ \boldsymbol{\rho}^{\mathrm{L4}},\ \boldsymbol{\rho}^{\mathrm{build}}\big)(\theta^\star),\qquad
R^{\mathrm{cert}}=\mathds{1}\big[V\ge 0\ \wedge\ \min\boldsymbol{\rho}^{\mathrm{L4}}\ge 0\ \wedge\ \min\boldsymbol{\rho}^{\mathrm{build}}\ge 0\big](\theta^\star),
\label{eq:m-l4}
\end{equation}
and $R^{\mathrm{cert}}$ enters Eq.~\ref{eq:sr}. Swapping the in-loop verdict $R$ for an LLM judge, all else fixed, quantifies the cost of misplaced acceptance authority (Section~\ref{sec:ablations}).

\begin{table}[!t]
\centering
\caption[Five systems $\times$ four LLMs on the main split]{Five systems $\times$ four LLMs on the topology-decision instances of the main split; effort \texttt{high}. Per block, \bestmark{best} and \secondmark{second}; tied values share a mark; lower is better for NFE@p, Exh., Wall, T/O, and Rank; ESC is unranked. $\pm$: Agresti--Coull SE (SEM for NFE@p). Wall in min, P/h in passes per hour, others in \%. PASS, Seis., Wind, Build, and Verif.\ divide by the 390 topology-decision answers per arm, ESC, Exh., T/O, and Valid by all 720 answers, and pass$^3$ by the 130 instances. Definitions and denominators in Appendix~\ref{app:exp}; THA: time-history analysis.}
\label{tab:m1}
\scriptsize
\setlength{\tabcolsep}{1.5pt}
\resizebox{\textwidth}{!}{%
\begin{tabular}{lrrrrrrrrrrrrrr}\toprule
 & \multicolumn{2}{c}{\textbf{Compliance}} & \multicolumn{4}{c}{\textbf{Efficiency}} & \multicolumn{1}{c}{\textbf{Reliability}} & \multicolumn{1}{c}{\textbf{Integrity}} & \multicolumn{2}{c}{\textbf{Held-out THA}} & \multicolumn{1}{c}{\textbf{Constr.}} & \multicolumn{3}{c}{\textbf{Composite}} \\
\cmidrule(lr){2-3}\cmidrule(lr){4-7}\cmidrule(lr){8-8}\cmidrule(lr){9-9}\cmidrule(lr){10-11}\cmidrule(lr){12-12}\cmidrule(lr){13-15}
\textbf{Arm} & \textbf{PASS} & \textbf{ESC} & \textbf{NFE@p} & \textbf{Exh.} & \textbf{Wall} & \textbf{T/O} & \textbf{pass$^3$} & \textbf{Valid} & \textbf{Seis.} & \textbf{Wind} & \textbf{Build} & \textbf{Verif.} & \textbf{P/h} & \textbf{Rank} \\ \midrule
\multicolumn{15}{c}{\textbf{DeepSeek-V4-Flash}} \\
Codex & 82.05$\pm$1.95 & 7.78$\pm$1.01 & \cellcolor{secondcolor}\underline{84.90}$\pm$1.40 & 59.72$\pm$1.82 & \cellcolor{secondcolor}\underline{6.3} & 6.94$\pm$0.96 & 66.15$\pm$4.10 & 91.39$\pm$1.06 & 56.92$\pm$2.50 & \cellcolor{secondcolor}\underline{64.87}$\pm$2.41 & 50.77$\pm$2.52 & 28.97$\pm$2.29 & \cellcolor{secondcolor}\underline{7.8} & 3.04 \\
Openclaw & 73.85$\pm$2.22 & 12.78$\pm$1.25 & 86.79$\pm$1.43 & \cellcolor{secondcolor}\underline{58.19}$\pm$1.83 & 11.0 & 7.78$\pm$1.01 & 48.46$\pm$4.32 & 93.75$\pm$0.92 & 52.31$\pm$2.52 & 59.23$\pm$2.48 & 47.18$\pm$2.52 & 28.97$\pm$2.29 & 4.0 & 4.29 \\
DeepSeek Harness & 82.82$\pm$1.91 & 6.81$\pm$0.95 & 98.08$\pm$0.60 & 91.81$\pm$1.03 & 8.6 & 7.22$\pm$0.98 & \cellcolor{secondcolor}\underline{67.69}$\pm$4.06 & \cellcolor{secondcolor}\underline{94.72}$\pm$0.85 & 52.82$\pm$2.52 & 64.10$\pm$2.42 & \cellcolor{bestcolor}\textbf{54.10}$\pm$2.51 & 28.72$\pm$2.28 & 5.8 & 3.50 \\
Claude Code & \cellcolor{secondcolor}\underline{83.85}$\pm$1.87 & 5.56$\pm$0.87 & 93.99$\pm$1.12 & 83.06$\pm$1.40 & 6.9 & \cellcolor{secondcolor}\underline{6.53}$\pm$0.93 & 66.92$\pm$4.08 & 94.03$\pm$0.90 & \cellcolor{secondcolor}\underline{58.72}$\pm$2.48 & 63.85$\pm$2.42 & \cellcolor{secondcolor}\underline{52.05}$\pm$2.52 & \cellcolor{secondcolor}\underline{31.03}$\pm$2.33 & 7.3 & \cellcolor{secondcolor}\underline{2.83} \\
\textbf{PE-Loop (ours)} & \cellcolor{bestcolor}\textbf{91.28}$\pm$1.45 & 11.39$\pm$1.19 & \cellcolor{bestcolor}\textbf{62.51}$\pm$0.38 & \cellcolor{bestcolor}\textbf{2.09}$\pm$0.31 & \cellcolor{bestcolor}\textbf{1.6} & \cellcolor{bestcolor}\textbf{1.67}$\pm$0.51 & \cellcolor{bestcolor}\textbf{68.46}$\pm$4.03 & \cellcolor{bestcolor}\textbf{98.89}$\pm$0.43 & \cellcolor{bestcolor}\textbf{60.77}$\pm$2.46 & \cellcolor{bestcolor}\textbf{70.00}$\pm$2.31 & 46.41$\pm$2.51 & \cellcolor{bestcolor}\textbf{35.64}$\pm$2.42 & \cellcolor{bestcolor}\textbf{34.2} & \cellcolor{bestcolor}\textbf{1.33} \\
\midrule \multicolumn{15}{c}{\textbf{Qwen3.8-Flash}} \\
Codex & 14.10$\pm$1.77 & 34.31$\pm$1.77 & \cellcolor{bestcolor}\textbf{28.24}$\pm$5.08 & \cellcolor{secondcolor}\underline{13.47}$\pm$1.28 & \cellcolor{secondcolor}\underline{13.4} & 18.89$\pm$1.46 & 2.31$\pm$1.64 & 56.39$\pm$1.84 & 13.33$\pm$1.73 & 14.10$\pm$1.77 & 12.56$\pm$1.69 & 12.05$\pm$1.66 & 0.6 & 3.83 \\
Openclaw & 32.56$\pm$2.36 & 1.11$\pm$0.43 & 98.72$\pm$0.50 & 42.92$\pm$1.84 & 30.0 & 56.11$\pm$1.84 & 7.69$\pm$2.47 & 54.31$\pm$1.85 & 20.00$\pm$2.03 & 24.62$\pm$2.18 & 20.77$\pm$2.05 & 10.51$\pm$1.57 & 0.7 & 4.50 \\
DeepSeek Harness & \cellcolor{secondcolor}\underline{65.90}$\pm$2.39 & 3.61$\pm$0.72 & 83.88$\pm$1.47 & 30.42$\pm$1.71 & 23.2 & 27.08$\pm$1.65 & \cellcolor{secondcolor}\underline{40.00}$\pm$4.24 & \cellcolor{secondcolor}\underline{77.50}$\pm$1.56 & \cellcolor{secondcolor}\underline{39.49}$\pm$2.46 & \cellcolor{secondcolor}\underline{50.00}$\pm$2.52 & \cellcolor{bestcolor}\textbf{42.56}$\pm$2.49 & \cellcolor{secondcolor}\underline{19.49}$\pm$2.01 & 1.7 & \cellcolor{secondcolor}\underline{2.67} \\
Claude Code & 56.15$\pm$2.50 & 12.50$\pm$1.24 & 80.25$\pm$1.57 & 17.08$\pm$1.40 & 17.1 & \cellcolor{secondcolor}\underline{16.39}$\pm$1.38 & 31.54$\pm$4.03 & 68.19$\pm$1.73 & 30.51$\pm$2.32 & 42.05$\pm$2.49 & \cellcolor{secondcolor}\underline{37.44}$\pm$2.44 & 12.82$\pm$1.71 & \cellcolor{secondcolor}\underline{2.0} & 2.75 \\
\textbf{PE-Loop (ours)} & \cellcolor{bestcolor}\textbf{77.44}$\pm$2.12 & 11.67$\pm$1.20 & \cellcolor{secondcolor}\underline{62.74}$\pm$0.38 & \cellcolor{bestcolor}\textbf{5.27}$\pm$0.19 & \cellcolor{bestcolor}\textbf{4.4} & \cellcolor{bestcolor}\textbf{1.53}$\pm$0.49 & \cellcolor{bestcolor}\textbf{63.85}$\pm$4.16 & \cellcolor{bestcolor}\textbf{98.89}$\pm$0.43 & \cellcolor{bestcolor}\textbf{58.21}$\pm$2.49 & \cellcolor{bestcolor}\textbf{68.72}$\pm$2.34 & 34.87$\pm$2.40 & \cellcolor{bestcolor}\textbf{23.59}$\pm$2.15 & \cellcolor{bestcolor}\textbf{10.7} & \cellcolor{bestcolor}\textbf{1.25} \\
\midrule \multicolumn{15}{c}{\textbf{GLM 5.2 Fast}} \\
Codex & \cellcolor{secondcolor}\underline{83.33}$\pm$1.89 & 2.50$\pm$0.61 & 96.86$\pm$0.78 & 80.56$\pm$1.48 & 11.0 & 10.56$\pm$1.15 & \cellcolor{secondcolor}\underline{64.62}$\pm$4.14 & 90.00$\pm$1.13 & \cellcolor{secondcolor}\underline{56.92}$\pm$2.50 & 63.33$\pm$2.43 & \cellcolor{secondcolor}\underline{52.05}$\pm$2.52 & \cellcolor{secondcolor}\underline{30.00}$\pm$2.31 & 4.6 & 2.92 \\
Openclaw & 77.95$\pm$2.10 & 2.78$\pm$0.64 & 98.95$\pm$0.31 & 79.44$\pm$1.51 & 18.4 & 15.69$\pm$1.36 & 56.15$\pm$4.29 & 89.72$\pm$1.14 & 53.33$\pm$2.51 & 63.85$\pm$2.42 & 50.77$\pm$2.52 & \cellcolor{bestcolor}\textbf{30.51}$\pm$2.32 & 2.5 & 3.58 \\
DeepSeek Harness & 69.49$\pm$2.32 & 0.56$\pm$0.34 & \cellcolor{secondcolor}\underline{85.31}$\pm$1.32 & \cellcolor{secondcolor}\underline{35.00}$\pm$1.77 & 25.0 & 26.39$\pm$1.64 & 38.46$\pm$4.21 & 86.94$\pm$1.26 & 44.87$\pm$2.51 & 50.77$\pm$2.52 & 50.00$\pm$2.52 & 23.85$\pm$2.16 & 1.7 & 4.42 \\
Claude Code & \cellcolor{bestcolor}\textbf{86.92}$\pm$1.72 & 1.53$\pm$0.49 & 90.85$\pm$0.93 & 45.42$\pm$1.85 & \cellcolor{secondcolor}\underline{8.5} & \cellcolor{secondcolor}\underline{5.97}$\pm$0.90 & \cellcolor{bestcolor}\textbf{80.77}$\pm$3.47 & \cellcolor{secondcolor}\underline{92.78}$\pm$0.98 & \cellcolor{secondcolor}\underline{56.92}$\pm$2.50 & \cellcolor{bestcolor}\textbf{67.18}$\pm$2.37 & \cellcolor{bestcolor}\textbf{56.41}$\pm$2.50 & 29.74$\pm$2.31 & \cellcolor{secondcolor}\underline{6.1} & \cellcolor{bestcolor}\textbf{2.00} \\
\textbf{PE-Loop (ours)} & 76.92$\pm$2.13 & 10.42$\pm$1.15 & \cellcolor{bestcolor}\textbf{63.31}$\pm$0.39 & \cellcolor{bestcolor}\textbf{6.41}$\pm$0.24 & \cellcolor{bestcolor}\textbf{2.5} & \cellcolor{bestcolor}\textbf{1.11}$\pm$0.43 & 63.85$\pm$4.16 & \cellcolor{bestcolor}\textbf{96.25}$\pm$0.73 & \cellcolor{bestcolor}\textbf{59.23}$\pm$2.48 & \cellcolor{secondcolor}\underline{66.92}$\pm$2.37 & 40.00$\pm$2.47 & 27.18$\pm$2.25 & \cellcolor{bestcolor}\textbf{18.3} & \cellcolor{secondcolor}\underline{2.08} \\
\midrule \multicolumn{15}{c}{\textbf{Kimi K3}} \\
Codex & 30.51$\pm$2.32 & 20.42$\pm$1.50 & \cellcolor{bestcolor}\textbf{56.71}$\pm$4.01 & 13.61$\pm$1.28 & \cellcolor{secondcolor}\underline{8.1} & \cellcolor{secondcolor}\underline{4.44}$\pm$0.79 & 11.54$\pm$2.88 & \cellcolor{secondcolor}\underline{86.39}$\pm$1.28 & 26.67$\pm$2.23 & 27.44$\pm$2.25 & 24.62$\pm$2.18 & \cellcolor{secondcolor}\underline{21.28}$\pm$2.07 & \cellcolor{secondcolor}\underline{2.2} & \cellcolor{secondcolor}\underline{2.50} \\
Openclaw & 5.13$\pm$1.16 & 0.42$\pm$0.31 & 95.00$\pm$2.57 & \cellcolor{secondcolor}\underline{2.50}$\pm$0.61 & 30.0 & 92.92$\pm$0.97 & 0.00$\pm$1.05 & 8.19$\pm$1.03 & 1.79$\pm$0.75 & 1.79$\pm$0.75 & 5.13$\pm$1.16 & 0.51$\pm$0.51 & 0.1 & 4.71 \\
DeepSeek Harness & 22.56$\pm$2.12 & 0.00$\pm$0.20 & 76.17$\pm$2.76 & 4.58$\pm$0.80 & 30.0 & 73.61$\pm$1.64 & 1.54$\pm$1.47 & 32.08$\pm$1.74 & 12.05$\pm$1.66 & 15.90$\pm$1.86 & 17.18$\pm$1.91 & 6.15$\pm$1.25 & 0.5 & 3.88 \\
Claude Code & \cellcolor{secondcolor}\underline{54.10}$\pm$2.51 & 0.69$\pm$0.36 & 82.36$\pm$1.56 & 17.92$\pm$1.43 & 21.2 & 29.31$\pm$1.69 & \cellcolor{secondcolor}\underline{24.62}$\pm$3.76 & 61.11$\pm$1.81 & \cellcolor{secondcolor}\underline{34.87}$\pm$2.40 & \cellcolor{secondcolor}\underline{39.49}$\pm$2.46 & \cellcolor{secondcolor}\underline{38.21}$\pm$2.45 & 18.21$\pm$1.96 & 1.5 & 2.83 \\
\textbf{PE-Loop (ours)} & \cellcolor{bestcolor}\textbf{83.08}$\pm$1.90 & 7.64$\pm$1.00 & \cellcolor{secondcolor}\underline{63.07}$\pm$0.33 & \cellcolor{bestcolor}\textbf{0.00}$\pm$0.19 & \cellcolor{bestcolor}\textbf{4.5} & \cellcolor{bestcolor}\textbf{2.50}$\pm$0.61 & \cellcolor{bestcolor}\textbf{74.62}$\pm$3.80 & \cellcolor{bestcolor}\textbf{97.64}$\pm$0.59 & \cellcolor{bestcolor}\textbf{61.54}$\pm$2.45 & \cellcolor{bestcolor}\textbf{71.03}$\pm$2.29 & \cellcolor{bestcolor}\textbf{43.33}$\pm$2.50 & \cellcolor{bestcolor}\textbf{28.21}$\pm$2.27 & \cellcolor{bestcolor}\textbf{11.0} & \cellcolor{bestcolor}\textbf{1.08} \\
\bottomrule\end{tabular}}
\end{table}

\begin{figure}[!t]
\centering
\includegraphics[width=\textwidth]{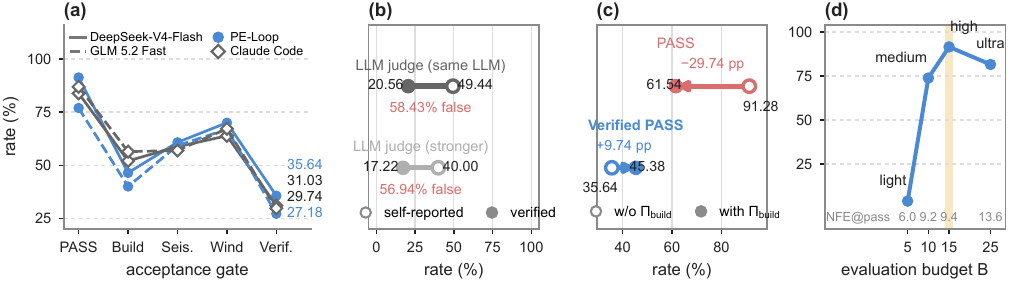}
\caption{Where certified success is won, and three interventions on PE-Loop. (a)~Gate survival of PE-Loop and the baseline leading raw PASS in both blocks, Claude Code, under DeepSeek-V4-Flash (solid) and GLM 5.2 Fast (dashed); all blocks in Fig.~\ref{fig:gates-all}. (b)~LLM judge in place of the margin oracle: self-reported (open) vs.\ physics-verified (filled) PASS (Table~\ref{tab:m3}). (c)~Rates without (open) and with (filled) $\Pi_{\mathrm{build}}$ (Table~\ref{tab:pibuild}). (d)~PASS against the evaluation budget (seed~0), with NFE@pass below.}
\label{fig:results}
\end{figure}

\section{Experiments}
\label{sec:experiments}
\paragraph{Setup.}
We evaluate five systems, namely PE-Loop, Codex, Openclaw, DeepSeek Harness, and Claude Code, each driven by four LLMs (DeepSeek-V4-Flash, Qwen3.8-Flash, GLM 5.2 Fast, and Kimi K3), on the 260-instance main split of VB-Bench (the selection procedure is described in Appendix~\ref{app:bench}). All rates are computed on the topology-decision instances of this split (all 130 for every system on every seed, Appendix~\ref{app:exp}) at effort level \texttt{high} ($B{=}15$, $K{=}3$) and averaged over three seeds. Only this stratum separates systems: section-only controls are passed by 93.6--97.9\% of answers under every system and boundary instances by none (Table~\ref{tab:strata}). Within each LLM block, the instances, the margin oracle (Eq.~\ref{eq:margin}), and the budget are identical, so the driving model is the only free variable. Verified PASS, defined as the conjunction of static compliance, the held-out time-history verification (Eq.~\ref{eq:m-l4}), and the constructability gate, serves as the empirical counterpart of $R^{\mathrm{cert}}$; column definitions and denominators are given in Appendix~\ref{app:exp}. Comparisons against LLM-free search use a separate search-sufficient reference set (GB-59, 59 instances $\times$ 3 seeds), which is never pooled with the main split. Of the 130 scored instances, 109 are text briefs, and 21 are sketch or image inputs; for the latter, the agent perceives $\hat M$, and adjudication runs on the true $M$. All rates pool the two groups.

\begin{figure}[!t]
\centering
\includegraphics[width=\textwidth]{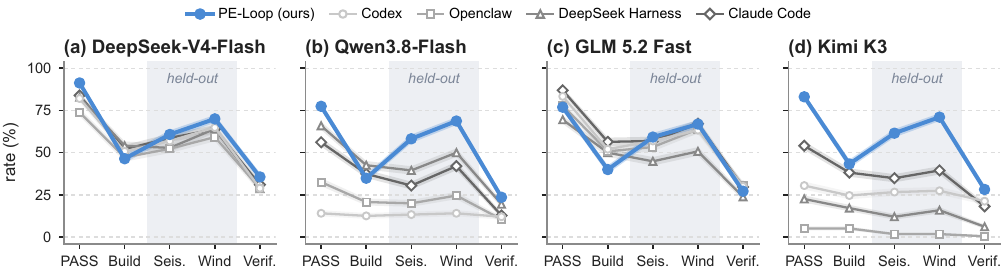}
\caption{Survival through the acceptance gates on all four LLM blocks (same construction as Fig.~\ref{fig:results}a; values from Table~\ref{tab:m1}; shaded band: the held-out re-checks).}
\label{fig:gates-all}
\end{figure}

\subsection{Main results}

\paragraph{Compliance tracks the verdict, not the model.}
PE-Loop holds 76.92--91.28\% PASS and 62.51--63.31\% NFE@pass across the four LLMs; only wall-clock time moves (1.6 to 4.5~min). A baseline is its own judge and inherits the model's variance: Codex falls from 83.33\% under GLM 5.2 Fast to 14.10\% under Qwen3.8-Flash. Section~\ref{sec:ablations} isolates the cause by swapping the verdict alone. The range across arms on raw PASS exceeds the 15~pp separation we pre-registered in every block (17.43~pp under DeepSeek-V4-Flash and GLM 5.2 Fast, the narrowest); the claims that follow nonetheless rest on the held-out and verified columns and on the four-backbone invariance, not on raw compliance alone.

\paragraph{Certified success is not raw compliance.}
Under GLM 5.2 Fast, the 10.00~pp raw lead of Claude Code (86.92 vs.\ 76.92) shrinks to 2.56~pp certified (29.74 vs.\ 27.18); under DeepSeek-V4-Flash, our 7.43~pp raw lead (91.28 vs.\ 83.85) shrinks to 4.61~pp certified (35.64 vs.\ 31.03), and the loss is concentrated at the constructability gate: 50.84\% of our statically compliant designs clear the detailing clause against 62.08\% for Claude Code, whereas at the held-out checks the two arms are within 4~pp of each other conditional on compliance (Fig.~\ref{fig:results}a; all blocks in Fig.~\ref{fig:gates-all}; Appendix~\ref{app:exp}). We hold the best Verified PASS in three of four blocks, the best seismic in all four, and the best wind in three; both exceptions, under GLM 5.2 Fast, lie within about one standard error of the difference (30.51 for Openclaw vs.\ 27.18 for PE-Loop on Verified PASS, $z\approx1.03$; 67.18 for Claude Code vs.\ 66.92 for PE-Loop on wind, $z\approx0.08$). Under DeepSeek-V4-Flash, our weakest gate is constructability (46.41\% vs.\ 47.18--54.10\%), diagnosed in Section~\ref{sec:ablations}; under Kimi K3, we already lead that gate.

\subsection{Ablations}
\label{sec:ablations}

\begin{figure}[!t]
\centering
\includegraphics[width=\textwidth]{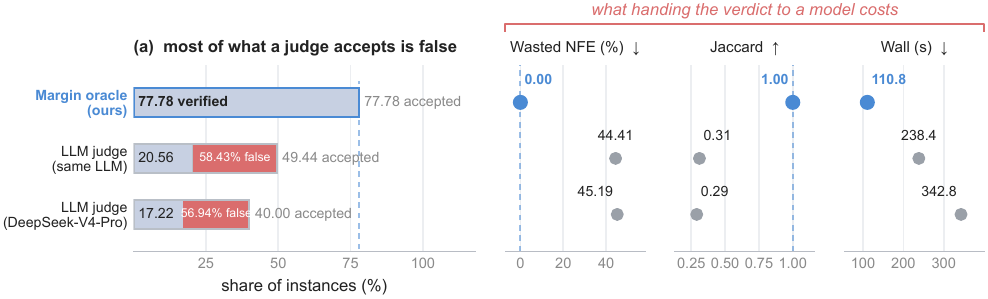}
\caption{Who holds the verdict (Table~\ref{tab:m3}). (a)~What each gate
accepts, split into physics-confirmed (light) and false (coral); dashed line:
what the oracle delivers. (b--d)~The price. Details in Appendix~\ref{app:exp}.}
\label{fig:judge}

\vspace{8pt}

\includegraphics[width=\textwidth]{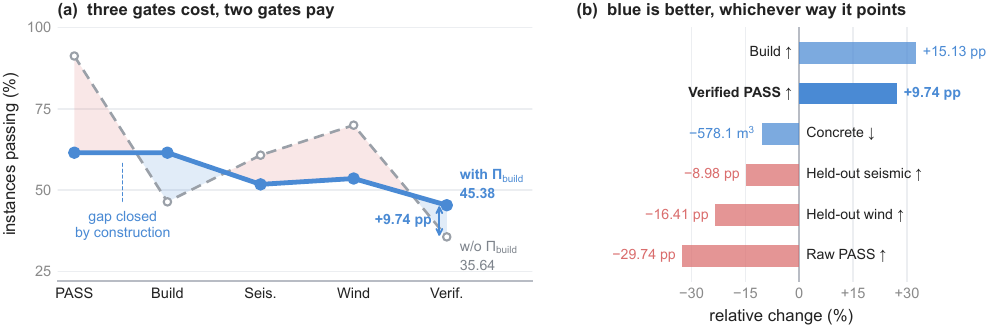}
\caption{$\Pi_{\mathrm{build}}$ trades raw for certified compliance
(Table~\ref{tab:pibuild}). (a)~Between the curves, coral is what it costs and
blue what it pays; (b)~relative change, blue is the better direction.
Details in Appendix~\ref{app:exp}.}
\label{fig:pibuild}
\end{figure}

\paragraph{Who holds the verdict (Fig.~\ref{fig:judge}; Table~\ref{tab:m3}).}
When $R$ (Eq.~\ref{eq:m-consistency}) is replaced by an LLM judge that sees the brief, the code limits, and the design parameters but no simulator output, with all else fixed, the same LLM claims 49.44\% PASS but delivers 20.56\%: 58.43\% of its acceptances are false, and 44.41\% of its evaluations are spent on misdirected repairs. A stronger, disjoint judge is no better (56.94\% false positives) and roughly triples wall-clock time, while the margin oracle claims exactly what it delivers (77.78\%; the 60 pre-registered instances hold 33 of the 34 failures of the main-table run, whose PE-Loop passes 81.67\% of them). Who holds the verdict drives compliance, not a correlate, consistent with the unreliability of unaided self-verification \citep{huang2024cannot,kambhampati2024llmmodulo} and with the optimization pressure a learned judge invites \citep{gao2023scaling}.

\paragraph{Constructability (Fig.~\ref{fig:pibuild}; Table~\ref{tab:pibuild}).}
Our constructability deficit comes from lever order: the outer loop enlarges the column before changing topology, and thick columns trip $l_n/h_b\ge4$ (Appendix~\ref{app:exp}). The gap is not a quirk of our loop: on the reference set, every system returns unbuildable designs, at 42.75--72.41\% of its feasible answers (Fig.~\ref{fig:build} in Appendix~\ref{app:build}). Moreover, 123 of our 126 violations are repairable by clipping one section scale, and $\Pi_{\mathrm{build}}$ does exactly that inside the loop: raw PASS falls from 91.28 to 61.54\% as thick-column designs that passed only by violating the clause are removed, but Verified PASS rises from 35.64 to 45.38\%, and concrete falls from 5,610.1 to 5,032.0~m$^3$.

\subsection{Budget and what the language model buys}

\begin{figure}[!t]
\centering
\includegraphics[width=\textwidth]{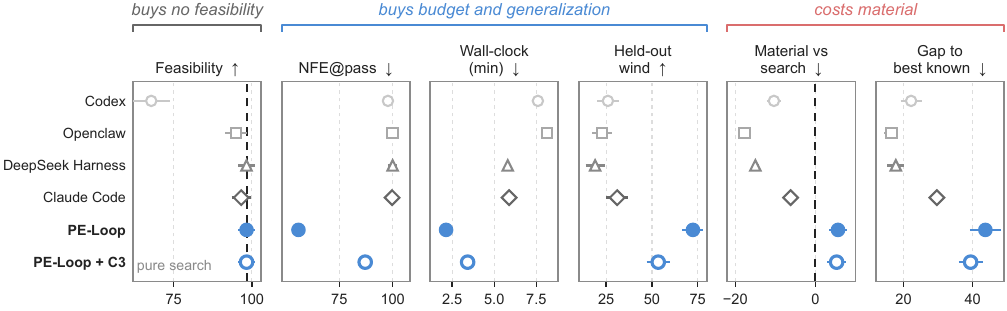}
\caption{What the language model buys on the search-sufficient reference set
(59 GB-family instances $\times$ 3 seeds, DeepSeek-V4-Flash; means $\pm$ SE).
Dashed line: exhaustive LLM-free section search. Details in Appendix~\ref{app:exp}.}
\label{fig:tracka}
\end{figure}

\paragraph{Budget discipline is structural.}
The escalation action makes the budget a property of the algorithm: because one outer round costs at most 9 oracle calls (Lemma~\ref{lem:bisect}), the exhaustion rate of PE-Loop never exceeds 6.41\% under any of the four LLMs (Table~\ref{tab:m1}). Under DeepSeek-V4-Flash, the baselines, which have no budget-aware escalation action, exhaust it on 58.19--91.81\% of their answers, and PE-Loop answers at 62.51\% NFE@pass compared with 84.90--98.08\%. The budget is also a dial, in the spirit of test-time compute scaling \citep{snell2025scaling}: on the identical 130 instances (seed~0), success runs 3.85 / 73.85 / 91.54\% from \texttt{light} to \texttt{high} (Fig.~\ref{fig:results}d), and \texttt{ultra} trades 10.00~pp of it for one more outer round (Appendix~\ref{app:exp}).

\paragraph{What the language model buys (Fig.~\ref{fig:tracka}).}
Against LLM-free pure search on the reference set, the feasibility of PE-Loop ties at 98.31\% with no discordant instances ($n_{01}{=}0$, $n_{10}{=}0$, $N{=}59$; McNemar's test, \citealp{mcnemar1947note}): the model buys no feasibility. It buys budget (NFE@pass 55.86\% vs.\ 97.78--100.00\%) and held-out generalization (72.41\% wind vs.\ 18.97--31.03\%). The price is material: 5.66\% more concrete than pure search, where every baseline spends less ($-17.74$ to $-6.17$\%), because $b^\star$ (Eq.~\ref{eq:m-bstar}) is minimal along a single coordinate; C3 recovers part of it, and closing the rest is the clearest open direction. Among the baselines, only Codex differs significantly from search on feasibility (Appendix~\ref{app:exp}).

\section{Conclusion and limitations}
\label{sec:limitations}
The results locate reliability in who holds the verdict, not in the proposer's strength. Once a deterministic engine alone accepts designs, held-out pass rates and oracle cost become properties of the loop: swapping the backbone barely moves them, whereas swapping only the verdict holder turns most acceptances into false positives, even with a stronger judge, because capability cannot substitute for access to the physics. Three implications follow. First, the language model's role narrows: on instances search alone can solve, it adds no feasibility; its value lies in choosing among discrete structural alternatives and spending a fixed budget well. Second, raw compliance is the wrong yardstick for physical-world agents: certified success reorders systems, and every system we tested, ours included, returns unbuildable designs that surface only under out-of-loop checks, so benchmarks should score certificates rather than claims. Third, a useful verifier is external, deterministic, and constraint-wise: its per-constraint margin both resists persuasion and tells the proposer what to repair. Limits remain: the held-out checks share the loop's physical model, certified success stays below one half, and single-coordinate search costs 5.66\% more concrete than exhaustive search; independent-solver validation, broader clause coverage, and multi-coordinate search come next. Wherever an accepted output is costly to reverse, the verdict should rest with an arbiter the proposer cannot persuade.

\subsection*{AI use statement}
We have not used generative AI tools to generate synthetic data, to develop the theoretical framework, to formulate or prove the claims of Lemmas~\ref{lem:bisect}--\ref{lem:safe}, to implement PE-Loop or the evaluation harness, to design the experiments, or to propose hypotheses; translation, dataset cleaning and reformatting, and qualitative data analysis do not apply to this work. We used generative AI tools to help interpret results: to cross-check the numbers reported in the text against the frozen result tables, to compute derived statistics (standard errors and $z$-scores of differences), and to draft the wording of result interpretations. We also used them to edit and condense text for readability, to create scientific figures and the scripts that generate them, and to write and edit verification code. We reviewed all AI-assisted work: we read back every plotted value from the figure objects and asserted it, the authors recomputed every reported number and derived statistic from the per-episode records, and the authors checked and approved all claims and interpretations. We take responsibility for the final content of this work, including text, claims, or artifacts produced with the aid of generative AI.

\subsection*{Ethics statement}
The designs returned by any system evaluated here are schematic and have not been reviewed by a licensed engineer; nothing in this paper substitutes for that review, and the benchmark certificate concerns compliance with the sourced clauses under the stated analyses, not a building's safety. The instances are derived from published building descriptions and public design codes and contain no personal data. The benchmark is intended to measure and improve the verifiability of automated design tools, which we expect to reduce, not increase, the risk that an unverified machine-generated design reaches construction.

\subsection*{Reproducibility statement}
We will release VB-Bench (instances, code-family limits, screening and difficulty labels, the 260-instance main split, and the 59-instance search-sufficient reference set), the margin oracle with the held-out seismic and wind analyses, PE-Loop, and the harness that drives all five systems upon publication. Every number in Section~\ref{sec:experiments} and Appendix~\ref{app:exp} is copied from frozen result tables produced by a single campaign; the release will include those tables, the per-episode records, the seeds, the model identifiers and effort tiers of Appendix~\ref{app:effort}, the denominators disclosed in Appendix~\ref{app:exp}, and the scripts that regenerate the figures and re-assert every number in the text against the tables. The oracle is deterministic; the residual variance comes from the language models, which we access through commercial APIs and therefore cannot reproduce bit-for-bit, so we report each rate over three seeds with its standard error.

\bibliography{iclr2027_conference}
\bibliographystyle{iclr2027_conference}

\appendix
\raggedbottom
\section{Related Work}
\label{sec:related}

\paragraph{Language model agents for structural analysis and design.}
Research on applying language models to structural analysis and design spans two levels. At the generation level, LLM agents translate natural language into executable finite-element code \citep{qi2025feagpt}, evolving toward lightweight, router-based multi-agent architectures for frame structures \citep{geng2025lightweight,geng2026agentic}, specialized member modeling \citep{wang2026helm}, code-compliant reinforced concrete design \citep{chen2025multiagent}, foundation type selection \citep{youwai2025router}, multiphysics solver generation \citep{zhang2025mooseagent}, retrieval-augmented structural calculation \citep{choi2026ramasc}, reliability-oriented structural analysis \citep{liu2026reliable}, and editable building models \citep{du2026text2bim}. At the verification level, multi-agent systems automate design review \citep{liang2025automating} and code compliance checking of building information models, increasingly refined through reinforcement learning \citep{madireddy2025bim,shi2026rl}, and benchmarks now grade domain knowledge \citep{liang2026aecbench}, simulation-validated engineering designs \citep{guo2025engdesign}, and the executable evidence chain from request to report \citep{qin2026structureclaw}. In most of these cases, the solver merely executes a task whose output is already known or verifiable. Hence, evaluation reduces to reproducing or checking a known answer rather than judging an undetermined design. Separately, data-driven methods generate shear-wall layouts with adversarial and diffusion models \citep{liao2021structgan,lu2022physicsgan,gu2024diffusion,liao2024genai} or combine generative neural fields, principal component analysis, and topology optimization \citep{sigmund2013topology} to produce diverse layouts from numerical objectives \citep{radler2025diverse,chen2024deep,yang2025ddtd}, without reading a natural-language brief. Concurrently, two systems close a repair loop around a deterministic structural solver \citep{heo2026hybrid,luo2026verification}. \citet{luo2026verification} let the LLM adjust member sections from the verifier's exceedance report on 44 member-scale cases (beams, trusses, and low-rise planar frames), report compliance rising from 56.8\% to 98.6\%, and attribute it to the verifier because it does not change detectably across two backbones; they note that the remaining failures stem from the initial topology, which a parameter-level loop cannot repair. We corroborate their attribution and extend it with topology-level repair by the language model while a deterministic search sets member sections, a benchmark of 3,577 building-scale instances across six code families, held-out seismic and wind re-checks that the loop never optimizes against, a counterfactual that replaces only the verdict holder (Table~\ref{tab:m3}), and a budget-bounded search with correctness guarantees (Lemmas~\ref{lem:bisect}--\ref{lem:safe}) compared against four coding harnesses under four backbones.

\paragraph{LLM agents grounded in physics and simulator feedback.}
Agents coupled with tools, physical models, or simulators interleave reasoning with action \citep{yao2023react} and explicitly plan and correct candidate decision paths using environment feedback or tree search \citep{yao2023tot}, turning one-shot generation into multi-step search and verification \citep{kim2025reflact,zhou2024lats,putta2024agentq}. Related work achieves policy improvement without parameter updates through structured self-reflection \citep{shinn2023reflexion}, distilling long-term interaction experience into reusable policy principles \citep{wang2024voyager,zhao2024expel} or structured skill libraries that evolve continually via reinforcement learning signals \citep{wu2025evolver,yang2026autoskill,wang2026sage}. Other work trains self-correction directly through multi-turn online reinforcement learning \citep{renze2024selfreflection,kumar2024score}. Who verifies matters: language-model judges carry systematic biases \citep{zheng2023judging,wang2024fair,gu2024llmasjudge}, models do not reliably correct their own reasoning without external feedback \citep{huang2024cannot}, refinement helps mainly when the critique is grounded in a tool \citep{madaan2023selfrefine,gou2024critic}, and policies optimized against a learned proxy overfit it \citep{gao2023scaling}, which motivates letting the model generate while sound external critics accept \citep{kambhampati2024llmmodulo}. Regarding feedback granularity, process supervision outperforms outcome supervision \citep{uesato2022process,lightman2024verify}, with simulator outputs, execution feedback, or process reward model returns used to score intermediate decisions and code-level proposals step by step \citep{gehring2025rlef,xi2025agentprm}; reinforcement learning from verifiable rewards trains against programmatic checkers \citep{shao2024deepseekmath,guo2025deepseekr1}, test-time computation can substitute for model scale when a reliable verifier exists \citep{snell2025scaling}, evaluator-in-the-loop search discovers new constructions and algorithms \citep{romeraparedes2024funsearch,novikov2025alphaevolve,lu2024aiscientist}, and executable benchmarks grade agents by running their outputs, from unit tests \citep{jimenez2024swebench} to repeated-trial reliability measured by pass$^k$ \citep{yao2025taubench}. A more demanding setting remains open, one in which an accepted proposal must satisfy several externally specified performance levels simultaneously, rather than depending on a single one-shot reward signal. This departs from building-design practice, where multiple codes and performance levels are checked in parallel, and motivates the closed-loop mechanism with multidimensional verification developed in this paper.

\section{VB-Bench details}
\label{app:bench}

\begin{figure}[!t]
\centering
\includegraphics[width=0.95\textwidth]{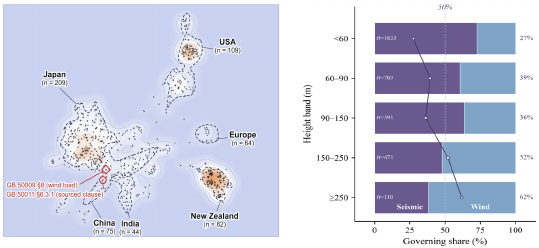}
\caption{\textbf{VB-Bench sources and governing loads.} (a) The structural-code knowledge base as one cross-linked graph. Each dot is a wiki entry (a code source document, an institution, or a clause concept) placed by a force-directed layout of the 1,847 wikilinks between entries; the shading is entry density weighted by link degree, from pale blue (sparse) through white to orange (hub pages). Dashed outlines enclose the densest region of each jurisdiction's entries, labeled with the number drawn; Europe's entries are the most dispersed, so its outline marks only its densest cluster. Two landmark pages are marked: GB~50011-2010 Clause~6.3.1, which sources the detailing rule of Table~\ref{tab:pibuild}, and the wind-load chapter of GB~50009-2012. 583 of the 625 entries (121 sources, 79 institutions, 425 concepts, distilled from 86 original codes) form the connected component shown; the remaining 42 unlinked or small-component pages are omitted. (b) Wind vs.\ seismic governing share by height band; wind governs the majority of instances above 150~m.}
\label{fig:benchapp}
\end{figure}

\paragraph{Instance and task layers.}
The instance layer provides 3,577 physically adjudicated buildings (Fig.~\ref{fig:bench}), grouped by construction route into synthetic instances (T1, 832), literature-grade reconstructions (T2, 2,010), real anchors (T3, 658), and complex-massing towers (T4, 77), which together supply both the source of design problems and a gradient of difficulty. The task layer contains 1,311 independently collected, gradable tasks, stratified by solvability and question type: intent (L1), design (L2), and modeling (L3). L2 spans the full set, while L1 (338) and L3 (384) are the subsets for which we could source intent and model ground truth. Contamination-resistant components comprise 1,197 paired variants and a quarantine set of 372 tasks.

\paragraph{From the instance layer to the experimental sets.}
The experiments do not run on the whole instance layer. Static screening (missing ground truth, too few stories, non-reinforced-concrete systems, infeasible slenderness) and a dynamic solvability calibration leave a difficulty map of 3,216 instances: 1,836 that a section change alone solves, 345 that require a topology decision, and 1,035 that no system solves (Table~\ref{tab:diffmap}); 32 further image-only instances are held out as a multimodal quarantine set and never used here, leaving a pool of 3,184. The 260-instance main split is a seeded, height-stratified sample of that pool: 130 topology-decision instances, 110 section-only controls, and 20 capability-boundary instances above 250~m, with per-band quotas set by supply: only five topology-decision instances exist between 120 and 250~m, so that band's quota was cut from 40 to 20. The freed slots went to the 24--60~m band, where topology decisions are densest. Running five systems on four backbones and three seeds over the full pool was outside the campaign's budget, and only the topology-decision stratum separates systems (Appendix~\ref{app:exp}). The smaller sets of Section~\ref{sec:experiments} derive from the same seed and nest inside one another: the 60-instance set of Table~\ref{tab:m3} was frozen together with its criteria before any run, and the 59-instance search-sufficient reference set is reused from an earlier campaign, whose exhaustive and grid-search reference lines make it the only set with an LLM-free baseline. PE-Loop and every baseline answered all 260 instances on every seed and backbone ($n{=}390$ topology-decision answers and $n{=}720$ answers in total per arm and backbone, the 20 boundary instances excluded from every comparative statistic; non-deliveries in Appendix~\ref{app:exp}). Table~\ref{tab:m1} reports scores only for the 130 topology-decision instances; the section-only and boundary strata are reported in Table~\ref{tab:strata}.

\paragraph{Acceptance.}
Every sample passes expert review, physical adjudication, and time-history verification. Acceptance uses the envelope of three load cases: gravity (self-weight) is always included, while seismic and wind are evaluated separately, and the most unfavorable value of each response quantity is taken without load-combination factors, since the envelope already provides a conservative upper bound. Acceptance limits scale with building height according to a tiered scale (Fig.~\ref{fig:benchapp}b).

\paragraph{Ground-motion and wind record libraries.}
Both held-out re-checks draw from fixed, public record libraries released with the benchmark (Fig.~\ref{fig:records}). The seismic library holds 62 recorded horizontal accelerograms in units of $g$: the 44 records of the FEMA P-695 far-field set~\citep{fema2009p695}, 22 shallow-crustal earthquakes of magnitude 6.5--7.6 recorded more than 10~km from the rupture, two horizontal components each, from San Fernando~1971 and Friuli~1976 to Chi-Chi~1999 and D\"uzce~1999, with peak accelerations of 0.21--0.82\,g, taken from the PEER NGA-West2 database~\citep{ancheta2014nga};\footnote{\url{https://ngawest2.berkeley.edu/}.} 16 further PEER records from six events (San Fernando 1971, Imperial Valley 1979, Manjil 1990, D\"uzce 1999, San Simeon 2003, and Chuetsu-oki 2007); and the two horizontal components of El Centro 1940 (Appendix~\ref{app:projectors}). Sampling intervals range from 0.0025 to 0.02~s and durations from 20 to 100~s; during ingestion, we checked every record against the point count and time step in its header and against its published peak acceleration, within 10\%. Only the P-695 subset is eligible for the held-out selection (Appendix~\ref{app:method}); we keep the rest for illustration. The wind library holds sixteen ten-minute wind-speed series at 5~Hz, stored as normalized fluctuations $v(t)/\bar v$ together with their turbulence intensity $I_u$. Eight are field records, resampled from the 20~Hz sonic-anemometer measurements at 110.5~m and 179.5~m on the tall meteorological tower of the IEA Wind Task~57 rotor-inflow benchmarks, four from stable and four from unstable atmospheric cases, with $I_u$ of 0.055--0.105~\citep{doubrawa2025rotorinflow,rybchuk2025rotorinflow};\footnote{Documentation: \url{https://rotor-inflow-benchmark.readthedocs.io/}. Data, CC BY 4.0: \url{https://zenodo.org/records/15747594}, \url{https://zenodo.org/records/15747684}.} eight are synthesized by spectral representation~\citep{shinozuka1972digital} from the Kaimal and Davenport spectra~\citep{kaimal1972spectral,davenport1961spectrum} at target $I_u$ of 0.12, 0.18, and 0.25, the three terrain categories the benchmark distinguishes, with fixed seeds.

\begin{figure}[!htb]
\centering
\includegraphics[width=0.95\textwidth]{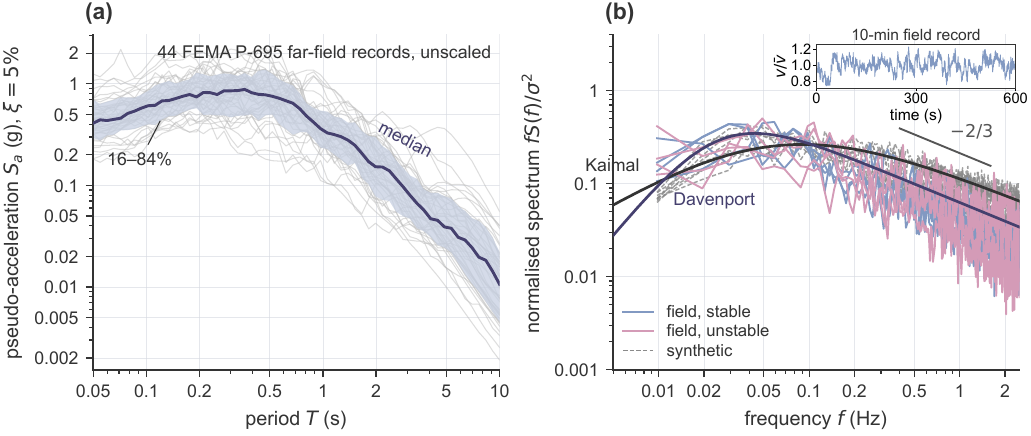}
\caption{\textbf{The held-out record libraries.} (a) Unscaled 5\%-damped pseudo-acceleration response spectra of the 44 FEMA P-695 far-field records (gray), with their median and 16--84\% band; only these records enter the held-out selection. (b) Normalized power spectra $fS(f)/\sigma^2$, with $\sigma^2$ the variance resolved between 0.01 and 2.5~Hz, of the eight field-recorded wind series (color, by atmospheric stability) and the eight spectrum-consistent synthetic series (gray, dashed), against the Kaimal and Davenport targets used for synthesis; the short line marks the $-2/3$ inertial-subrange slope; inset, one field record over its ten minutes. Sources and processing in the text.}
\label{fig:records}
\end{figure}

\paragraph{Structural-code knowledge base.}
To support clause-level verification of design limits, we build a parallel knowledge base. It collects 86 original structural design codes from six countries and regions, namely China (GB~50011-2010, JGJ~3-2010, and GB~50009-2012; \citealp{gb50011,jgj3,gb50009}), Europe (Eurocode~8; \citealp{en1998}), the United States (ASCE/SEI~7; \citealp{asce7}), India (IS~1893; \citealp{is1893}), Japan (the Building Standard Law, BSL; \citealp{bsl_japan}), and New Zealand (NZS~1170.5; \citealp{nzs11705}), and organizes them into 625 cross-linked entries comprising 121 code source documents, 79 institution and committee entities, and 425 clause concepts (Fig.~\ref{fig:benchapp}a). It covers seismic and wind design provisions, load values (dead, live, wind, and seismic), and concrete structural design rules, and tracks each code's revision history. Every limit in the task set carries a clause number traceable to this knowledge base.

\paragraph{Detailing set.}
Every code-compliance instance carries a detailing set $G^{\mathrm{build}}_x$ drawn from its code family with clause provenance in the knowledge base. The current release sources one clause family, the beam clear-span-to-depth limit $l_n/h_b\ge 4$ (GB~50011-2010 Clause~6.3.1, \citealp{gb50011}; JGJ~3-2010 Clause~6.3.1, \citealp{jgj3}), evaluated geometrically from the returned member sections at no simulator cost; the two further candidate clauses and the reasons for their exclusion are given in Appendix~\ref{app:build}.

\section{One massing, four projectors}
\label{app:projectors}

Figures~\ref{fig:app-classics} and~\ref{fig:app-supertall} show six buildings, each read from one massing by four tools, to make visible what the margin oracle of Section~\ref{sec:method} analyzes. The four columns are the same in both figures. They are headed by the tool that produced them: the massing in Rhino\footnote{Rhinoceros 3D, Robert McNeel \& Associates, \url{https://www.rhino3d.com}.}, a Blender\footnote{Blender, Blender Foundation, \url{https://www.blender.org}.} rendering of it, the frame idealization the physics engine adjudicates in OpenSeesPy~\citep{mckenna2011opensees,zhu2018openseespy}, and that frame's response to a recorded ground motion. The first two columns are what a designer sees; the last two are what the physics engine sees. Table~\ref{tab:app-projectors} lists the six frame models.

\paragraph{Rhino (column 1).} For the three classics, the panel is a still from a four-viewport screen recording of the interactive Rhino session in which the massing was modeled; it is the first frame in which the model is complete. The benchmark's LLM-free search produced the three towers, so no such session exists, and our headless Rhino export writes box massings only; their panel therefore draws the massing the frame generator starts from, one lofted perimeter ring per floor, in a Rhino-style four-viewport layout. That drawing is ours, not a Rhino screen.

\paragraph{Blender (column 2).} For the classics, the panel shows the Blender scene of the same massing, re-rendered headless with Cycles on a light world. For the towers, we loft the column~1 massing, mesh it, and render it with the same world, camera rule, and Cycles settings, with a glazed facade and a story band per floor so the form reads as a building.

\paragraph{OpenSeesPy model (column 3).} The panel shows the frame idealization that the benchmark generates and the physics engine analyzes, as four orthographic views (plan, north and east elevations, isometric). Every structural member is drawn; in the tower rows, facade members are hidden so that columns, cores, and belts are visible.

\paragraph{OpenSeesPy response (column 4).} The panel shows the deformed shape of the column-3 model at the instant of peak roof drift under a linear time-history analysis run for these figures: the east--west (S90W) component of the 1940 Imperial Valley earthquake recorded at the El Centro terminal substation, USGS station 117, in the Caltech corrected digitization~\citep{hudson1971strong} as redistributed by vibrationdata.com\footnote{\url{https://www.vibrationdata.com/elcentro_EW.dat}, two-column file of time and acceleration in $g$.} (PGA 0.214\,g, 2,674 points at $\Delta t=0.02$\,s, 53.5\,s), uniform base excitation along global $X$, Newmark average acceleration~\citep{newmark1959method} with the effective stiffness factorized once, Rayleigh damping of 5\% on the first two modes, every node recorded at every step. Members are colored by $|u|$ from zero to that row's own peak, against the undeformed ghost. This analysis is illustrative: it is not the held-out verification of Eq.~\ref{eq:m-l4}, which uses the benchmark's ground-motion set and is scored, not drawn. Before any step was run, each rebuilt model was asserted to reproduce the node count, element count, and first period stored with the building; Table~\ref{tab:app-projectors} gives those counts, $T_1$, the number of structurally connected pieces the frame generator leaves, the peak roof drift, and the peak $|u|$.

\paragraph{What these figures do not claim.} The classics are massings modeled in an interactive Rhino session before this work and turned into frames by the same generator; the towers come from the benchmark's LLM-free search (target derivation, grid search over sections, verification). None of the six is an episode of Section~\ref{sec:experiments}. The three mega-frames are certified on drift only, since the mega-frame model does not evaluate top acceleration, period ratio, or the shear-weight and stiffness-weight ratios; the classics carry every check. The response model is screening-level and elastic: linear members, no gravity in the seismic case, no material nonlinearity, no soil, one horizontal component of one record. We did not compare against any measured or published response of these buildings, and the figures make no claim about the analysis's accuracy. They show that one massing consistently drives four tools and what the last two contain.

\begin{table}[!htb]
\centering\small
\caption{The six frame models of Figs.~\ref{fig:app-classics} and~\ref{fig:app-supertall}: size of the idealization that the physics engine analyzes, first period $T_1$, structurally connected pieces left by the frame generator, and peak response under El Centro 1940 EW (settings in Appendix~\ref{app:projectors}).}
\label{tab:app-projectors}
\resizebox{\textwidth}{!}{%
\begin{tabular}{lrrrrrrrr}
\toprule
Massing (after) & $H$ (m) & Stories & Nodes & Elements & $T_1$ (s) & Pieces & Peak roof drift & Peak $|u|$ (mm) \\
\midrule
\multicolumn{9}{l}{\emph{Classics (Fig.~\ref{fig:app-classics}), displacements drawn $\times$50}} \\
Seagram Building & 157 & 38 & 2,746 & 7,469 & 2.55 & 2 & $H/483$ & 325 \\
Hilversum town hall & 48 & 13 & 243 & 454 & 0.92 & 4 & $H/503$ & 96 \\
Casa del Fascio & 16.6 & 4 & 280 & 576 & 0.49 & 1 & $H/456$ & 42 \\
\midrule
\multicolumn{9}{l}{\emph{Supertall towers (Fig.~\ref{fig:app-supertall}), displacements drawn $\times$100}} \\
Burj Khalifa & 828 & 162 & 6,194 & 20,736 & 8.44 & 1 & $H/697$ & 1,189 \\
Shanghai Tower & 632 & 140 & 5,358 & 17,964 & 8.57 & 1 & $H/828$ & 762 \\
Taipei 101 & 508 & 101 & 3,264 & 10,096 & 4.35 & 1 & $H/807$ & 630 \\
\bottomrule
\end{tabular}}
\end{table}

\begin{figure}[!tp]
\centering
\includegraphics[width=\textwidth]{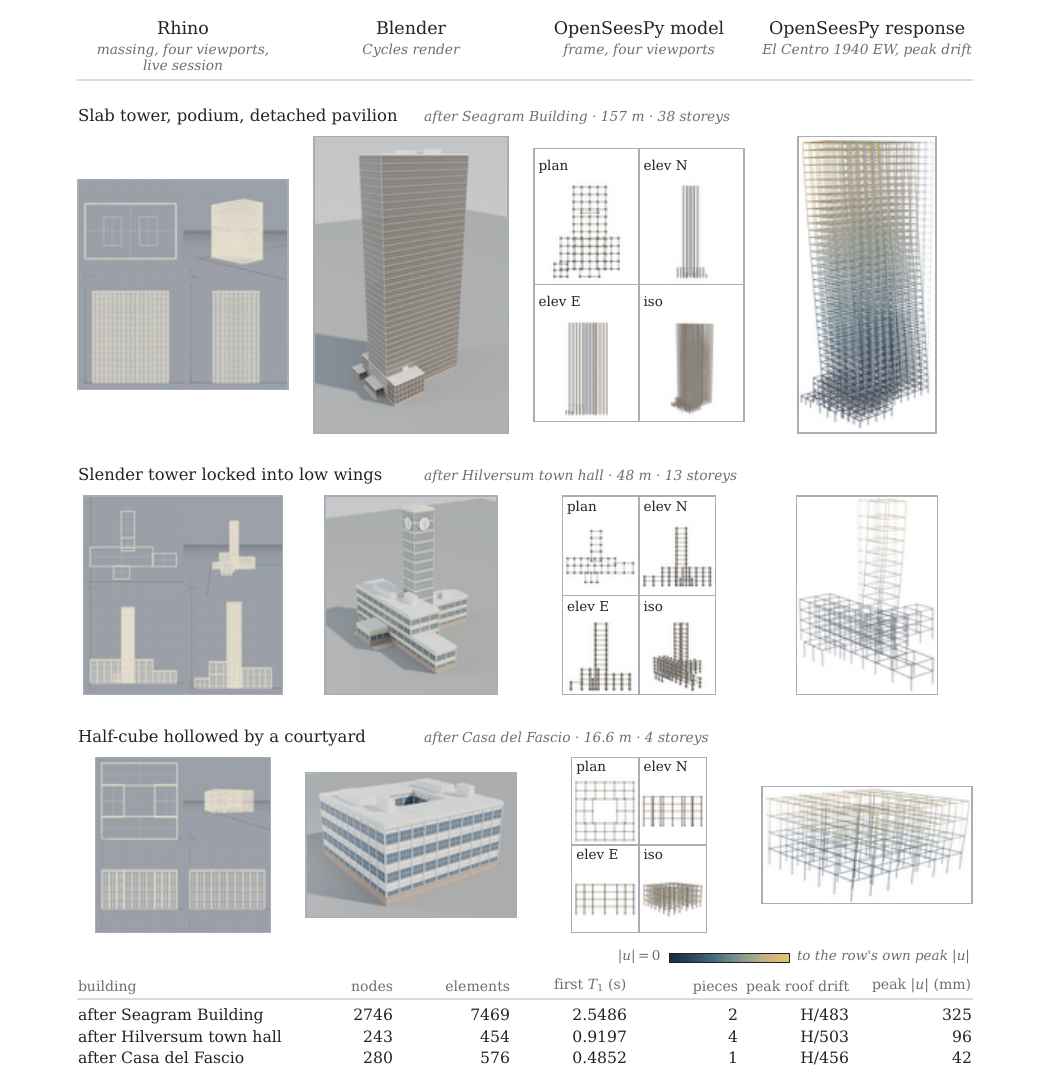}
\caption{\textbf{Classic buildings, four projectors.} Three massings modeled in an interactive Rhino session before this work: a slab tower on a podium with a detached pavilion (after the Seagram Building, 157\,m, 38 stories), a slender tower locked into low wings (after the Hilversum town hall, 48\,m, 13 stories), and a courtyard block (after the Casa del Fascio, 16.6\,m, 4 stories). Column~1 is a still of that session, column~2 its Blender scene, columns~3 and~4 the frame the physics engine analyzes and its peak response, displacements drawn $\times$50. The town hall's four structurally separate pieces move independently in column~4; the frame generator leaves this, not a display fault. Columns and analysis as described in Appendix~\ref{app:projectors}; model sizes and peak responses in Table~\ref{tab:app-projectors}.}
\label{fig:app-classics}
\end{figure}

\begin{figure}[!tp]
\centering
\includegraphics[width=\textwidth]{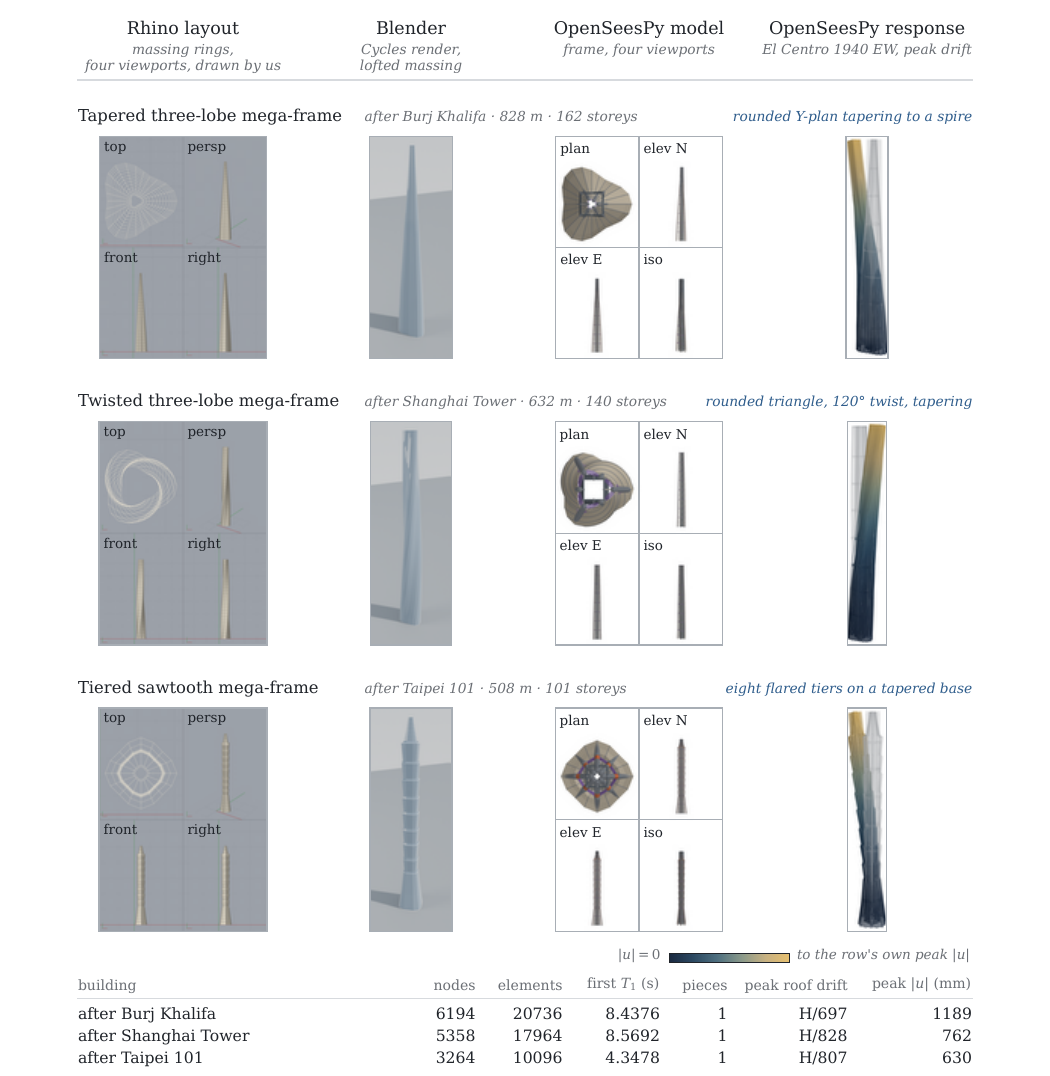}
\caption{\textbf{Supertall towers, four projectors.} Three mega-frames from the benchmark's LLM-free search: a rounded Y-plan tapering to a spire (after Burj Khalifa, 828\,m, 162 stories), a rounded triangle with a 120$^\circ$ twist (after Shanghai Tower, 632\,m, 140 stories), and eight flared tiers on a tapered base (after Taipei~101, 508\,m, 101 stories). Column~1 is the lofted massing we drew, column~2 is our Cycles render, and columns~3 and~4 show the frame the physics engine analyzes and its peak response, with displacements drawn $\times$100. Seen from above, the ring beams of a tapering or twisting tower sit at a different radius on every floor, so the plan viewports of columns~1 and~3 both show a stacked disc. All three are drift-certified only. Columns and analysis are as described in Appendix~\ref{app:projectors}; model sizes and peak responses are in Table~\ref{tab:app-projectors}.}
\label{fig:app-supertall}
\end{figure}

\section{Method details and derivations}
\label{app:method}

\subsection{In-loop constraint set and the role of the PPR}
\label{app:signal}

The in-loop constraint set $G_x$ is fixed by the code family and the task type of the instance (Section~\ref{sec:task}) and is evaluated by equivalent static and modal analysis. Code-compliance instances use the checks of the bound code family for which limits could be sourced, typically upper bounds on the inter-story drift ratio, the torsional displacement ratio, the minimum shear-to-weight ratio, the story lateral-stiffness ratio, the torsional-to-translational period ratio, and the stiffness-to-weight ratio (global stability index); performance-target instances additionally carry a lower bound on the fundamental frequency and an upper bound on the static top displacement. For buildings taller than 150~m, an upper bound on wind-induced peak acceleration also enters $G_x$ and is checked at every evaluation by a modal-superposition wind time history. Response-history analysis under recorded ground motions never enters $G_x$; candidates rejected by the zero-reference gates of Section~\ref{sec:task} are treated as $R=0$ with every constraint violated.

The PPR is formally close to potential-based reward shaping~\citep{ng1999policy}. However, it is neither a hand-designed potential nor does it enter the return in the difference form $\gamma\phi(s')-\phi(s)$; it is used only as an observation and a verdict (violated constraints guide repair before feasibility, concrete volume ranks candidates after feasibility). It therefore inherits no policy-invariance guarantee, and such theoretical properties are outside the scope of this work. Because the design loop consumes only the sign of $\boldsymbol{\rho}$, its behavior does not depend on any heuristic weighting of margin magnitudes, and there is no tunable intermediate quantity between guidance and verdict. Magnitudes are returned to the agent and recorded in the certificate but not consumed by the search; using them, e.g., to order C1 actions by the most violated constraint, is a natural extension we leave open.

\subsection{Initial topology and the column grid}
\label{app:intent}

In $\theta_d=(n_{\mathrm{core}},n_{\mathrm{wall}},\kappa,\varsigma)$, the core plan geometry $\kappa$ consists of the core-to-floor area ratio, the plan aspect ratio, and the plan offset; the section triple $\varsigma$ consists of the beam-section scale, the concrete-grade scale, and the vertical taper, each on a small finite menu. The intent $\hat y=(\text{system},\text{governing load},\text{strategy})=\pi_{\mathrm{int}}(\hat M,S)$ fixes the initial topology
\begin{equation}
\begin{aligned}
n^{(0)}_{\mathrm{core}}&=\mathds{1}\big[\texttt{add\_core}\in\text{strategy}\ \vee\ \text{system}=\text{frame-core}\big],\\
n^{(0)}_{\mathrm{wall}}&=\mathds{1}\big[\texttt{add\_wall}\in\text{strategy}\ \vee\ \text{system}=\text{frame-shear-wall}\big];
\end{aligned}
\label{eq:m-intent}
\end{equation}
$\kappa^{(0)}$ and $\varsigma^{(0)}$ take the defaults (area ratio $0.15$, aspect ratio $1$, zero offset; beam-section and concrete-grade scales $1$, taper $0$). Equation~\ref{eq:m-intent} writes the semantic decision explicitly into $\theta_d^{(0)}$, which also lets Section~\ref{sec:experiments} test whether the final design is consistent with the declared intent.

The column grid $G_b$ of Eq.~\ref{eq:m-space} is a convention rather than a derivation, and we state its provenance. The 50~mm step is the construction modulus on which reinforced-concrete column dimensions are specified, so every $b\in G_b$ is buildable. The lower end, 0.30~m, is the smallest section the parametric frame template instantiates. The upper end, 2.50~m, is the clamp on column sections shared by the search spaces of all compared systems; an earlier, agent-specific ceiling of 1.6~m made the agent escalate before the shared search space was exhausted, so we unified it to the shared clamp. The resulting size $|G_b|=45$ is what fixes the bisection cost in Lemma~\ref{lem:bisect}.

\subsection{Detailing set and projection}
\label{app:build}

The frame template sizes members from $\theta$ alone: the column is square with side $b$; the beam depth and width are $h_b=0.8\,b\,\varsigma_{\mathrm{beam}}$ and $b_b=\max(0.4\,b,\ h_b/2.5)$ with $\varsigma_{\mathrm{beam}}\in[0.5,2.0]$; the bay width in each plan direction is $W/\max(1,\lfloor W/6\,\mathrm{m}\rceil)$, with $\lfloor\cdot\rceil$ rounding to the nearest integer (and likewise for $D$), so the beam clear span is $l_n=\mathrm{bay}-b$, taken as the smaller of the two directions. The detailing set $G^{\mathrm{build}}_x$ currently sources one clause family, the clear-span-to-depth limit $l_n/h_b\ge 4$ (GB~50011-2010 Clause~6.3.1; JGJ~3-2010 Clause~6.3.1), whose margin follows Eq.~\ref{eq:margin} with $q=h_b/l_n$, $\mathrm{lim}=1/4$, and $s=+1$. Because $h_b$ grows with $b$ while $l_n$ shrinks with $b$, the constraint binds for thick columns: with $\mathrm{bay}=6$~m and $\varsigma_{\mathrm{beam}}=1$ it is violated for $b>6/4.2=1.43$~m, i.e., from $b=1.45$~m on the grid.

Before every evaluation the agent applies the projection
\begin{equation}
\Pi_{\mathrm{build}}:\quad
\varsigma_{\mathrm{beam}}\leftarrow\min\{\varsigma_{\mathrm{beam}},\ l_n/(3.2\,b)\},\qquad
b_{\max}\leftarrow\max\{\,b\in G_b:\ b\le \mathrm{bay}/2.6\,\},
\label{eq:m-proj}
\end{equation}
which follows directly from the clause. Solving $l_n/(0.8\,b\,\varsigma_{\mathrm{beam}})\ge 4$ for the beam scale gives $\varsigma_{\mathrm{beam}}\le l_n/(3.2\,b)$; requiring this bound to admit the smallest scale on the menu, $\varsigma_{\mathrm{beam}}\ge 0.5$, gives $\mathrm{bay}-b\ge 1.6\,b$, i.e., $b\le\mathrm{bay}/2.6$, which caps the grid at $b_{\max}$. Both operations are analytic in $\theta$ and cost no oracle calls; because Assumption~\ref{ass:mono} is stated on the projected path, Lemma~\ref{lem:bisect} applies to the capped grid verbatim (Appendix~\ref{app:c2}). Applied post hoc to designs returned without the projection, the same clipping repairs every violation except those with $b>\mathrm{bay}/2.6$, which no beam scale can rescue.

Figure~\ref{fig:build} is the measurement that motivated placing the projection inside the loop rather than applying it after the fact. For each compared system, it reports the share of physically feasible returned designs on the 59-instance search-sufficient reference set, pooled over three seeds, that violate the clear-span-to-depth limit. The gap grows with column size and is largest for the agent without the projection, because its bisection settles on the thickest feasible column. All but three of the violations are repairable by clipping the beam scale alone; the three exceptions have $b>\mathrm{bay}/2.6$.

We analyzed two additional clause families that are not sourced. The joint rule that the beam must not be wider than the column (GB~50011-2010 Clause~6.3.2 otherwise requires flat-beam provisions) can bind only in a corner of the menu: $b_b/b\le\max(0.4,\,0.32\,\varsigma_{\mathrm{beam}})$ while the tapered top-column side is $(1-\text{taper})\,b\ge 0.6\,b$. Hence, a violation requires $0.32\,\varsigma_{\mathrm{beam}}>1-\text{taper}$, i.e., $\varsigma_{\mathrm{beam}}>1.875$ at the maximum taper of $0.40$. No returned design reaches this corner (Appendix~\ref{app:exp}), so the rule is monitored but not sourced as a constraint. The seismic-grade-dependent minimum column side (300~mm or 400~mm, GB~50011-2010 Clause~6.3.5; JGJ~3-2010 Clause~6.4.1) requires the seismic grade of GB~50011-2010 Table~6.1.2, which is not defined uniformly across the code families of the benchmark; the axial-compression-ratio limit requires member forces that the harness does not expose. Both are reported as limitations.

\begin{figure}[!t]
\centering
\includegraphics[width=0.95\textwidth]{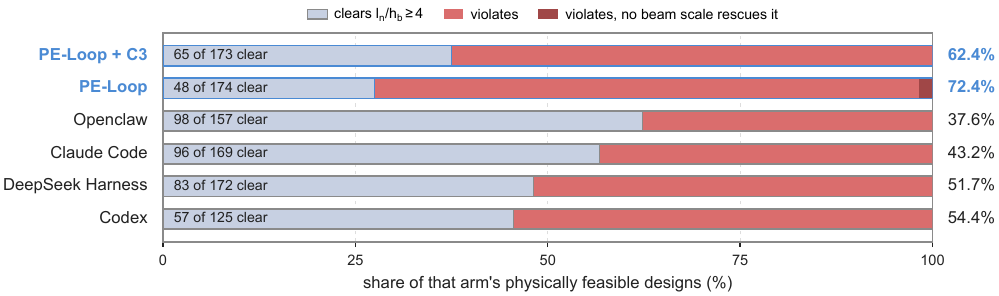}
\caption{The detailing gap is not ours alone. Each bar splits one arm's feasible designs in the 59-instance reference set by the clause $l_n/h_b\ge4$; the right rate is the violating share. Every system returns unbuildable designs (42.75--72.41\%). PE-Loop without $\Pi_{\mathrm{build}}$ is the worst, but 123 of its 126 violations are repairable by clipping one section scale, which $\Pi_{\mathrm{build}}$ does by construction (Table~\ref{tab:pibuild}). Coverage in Table~\ref{tab:gb59cov}.}
\label{fig:build}
\end{figure}

\subsection{Pseudocode}
\label{app:algo}

Algorithm~\ref{alg:peloop} gives the episode-level control flow of PE-Loop as implemented, in the notation of Sections~\ref{sec:task} and~\ref{sec:method}; the color convention matches the hatching of the method figure: \llm{blue} marks the three language-model calls, which are not charged to the budget, and \textcolor{physred}{red} marks the margin-oracle calls, each of which costs one evaluation, so a round's evaluation count can be read off the red lines directly. In the $\Pi_{\mathrm{build}}$ variant, the projection of Eq.~\ref{eq:m-proj} is applied to every candidate before evaluation; the harness-side re-verification of the returned design on the true massing $M$ (Section~\ref{sec:task}) and the offline final verification L4 (Appendix~\ref{app:l4}) are not part of the agent and are omitted.

\begin{algorithm}[!t]
\caption{PE-Loop episode. \llm{Blue}: language-model calls, not charged to $B$. \textcolor{physred}{Red}: margin-oracle calls (Eq.~\ref{eq:margin}), each charged one evaluation. In the $\Pi_{\mathrm{build}}$ variant, every oracle call is preceded by the projection of Eq.~\ref{eq:m-proj}. The harness re-verification on $M$ and the held-out L4 analyses are not part of the agent.}
\label{alg:peloop}
\begin{algorithmic}[1]
\Require instance $x=(O_\tau,S)$; grid $G_b=\{b_0<\dots<b_{44}\}$ (Eq.~\ref{eq:m-space}); action set $\mathcal{A}$; budget $B$; rounds $K$; tier $\in\{\texttt{light},\dots,\texttt{ultra}\}$
\Ensure status $o\in\{\textsc{pass},\textsc{escalated}\}$; design $\theta^\star$; margins $\boldsymbol{\rho}(\theta^\star)$; evaluations $\mathrm{NFE}$
\phase{B0--B2 \ perception and intent (Eq.~\ref{eq:m-intent})}
\State $\hat M\gets\llm{\pi_{\mathrm{perc}}(O_\tau)}$ if $\tau\neq\mathrm{text}$, else parse the brief
\State $\hat y\gets\llm{\pi_{\mathrm{int}}(\hat M,S)}$;\quad $\theta_d^{(0)}\gets$ initial state of Eq.~\ref{eq:m-intent};\quad $\mathrm{NFE}\gets 0$
\For{$k=0,1,\dots,K$}
  \phase{C2 \ minimal feasible column side under $\theta_d^{(k)}$ (Eq.~\ref{eq:m-bstar}; Lemmas~\ref{lem:bisect}--\ref{lem:cost})}
  \State $b_{\max}\gets\max\{b\in G_b:\ b\le\mathrm{bay}/2.6\}$ if $\Pi_{\mathrm{build}}$, else $b_{44}$ \Comment{$\Pi_{\mathrm{build}}$ caps the grid}
  \State\label{ln:probe} \Eval{\theta_d^{(k)},b_{\max}} \Comment{upper-bound probe}
  \If{$R(\theta_d^{(k)},b_{\max})=1$}
    \State $\ell\gets 0$;\ $h\gets\operatorname{index}(b_{\max})$
    \While{$\ell<h$ \textbf{and} $\mathrm{NFE}<B-1$} \Comment{$\le\lceil\log_2|G_b|\rceil=6$ steps}
      \State\label{ln:bisect} $m\gets\lfloor(\ell+h)/2\rfloor$;\ \Eval{\theta_d^{(k)},b_m}
      \State \textbf{if} $R(\theta_d^{(k)},b_m)=1$ \textbf{then} $h\gets m$ \textbf{else} $\ell\gets m+1$
    \EndWhile
    \State\label{ln:confirm} \Eval{\theta_d^{(k)},b_\ell} \Comment{confirmation; catches budget truncation}
    \If{$R(\theta_d^{(k)},b_\ell)=1$}
      \State $\theta^\star\gets(\theta_d^{(k)},b_\ell)$;\ $o\gets\textsc{pass}$;\ \textbf{break} \Comment{\textsc{pass} before exhaustion}
    \EndIf
  \EndIf
  \State $\mathcal{J}\gets\{\,j:\ \rho_j<0\,\}$ of the last evaluation \Comment{Eq.~\ref{eq:m-violated}; the only signal C1 receives}
  \phase{C1 \ one discrete repair, proposed by the language model (Eq.~\ref{eq:m-outer})}
  \If{$k=K$ \textbf{or} $\mathrm{NFE}\ge B$}
    \State $\theta^\star\gets(\theta_d^{(k)},b_{\max})$;\ $o\gets\textsc{escalated}$;\ \textbf{break}
  \EndIf
  \State $a_k\gets\llm{\pi_{\mathrm{C1}}(\hat M,\theta_d^{(k)},\mathcal{J})}\in\mathcal{A}$;\quad $\theta_d^{(k+1)}\gets a_k(\theta_d^{(k)})$ \Comment{clipped to admissible ranges}
\EndFor
\phase{C3 \ material loop along $\varsigma$, \textup{\texttt{ultra}} tier only (Eq.~\ref{eq:m-c3})}
\If{$o=\textsc{pass}$ \textbf{and} tier $=\texttt{ultra}$}
  \ForAll{$\theta'$ in the candidate menu, while $\mathrm{NFE}<B$}
    \State \Eval{\theta'};\ \textbf{if} $R(\theta')=1\ \wedge\ c(\theta')<c(\theta^\star)$ \textbf{then} $\theta^\star\gets\theta'$
  \EndFor
\EndIf
\If{$\mathrm{NFE}<B$}
  \State \Eval{\theta^\star} \Comment{final measurement; L4 and $\mathcal{C}$ offline (Eq.~\ref{eq:m-l4})}
\EndIf
\State \Return $o,\ \theta^\star,\ \boldsymbol{\rho}(\theta^\star),\ \mathrm{NFE}$
\end{algorithmic}
\end{algorithm}

\subsection{Correctness of C2 and its evaluation bound}
\label{app:c2}

\begin{assumption}[Monotone margin along the projected column axis]
\label{ass:mono}
For fixed $\theta_d$ and $b\le b'$ in $G_b$, $V(\Pi_{\mathrm{build}}(\theta_d,b))\le V(\Pi_{\mathrm{build}}(\theta_d,b'))$.
\end{assumption}

Assumption~\ref{ass:mono} is stated on the projected path because $\Pi_{\mathrm{build}}$ may lower the beam scale as $b$ grows. Two facts confine its role. First, the beam clip is inactive on the whole grid whenever $\varsigma_{\mathrm{beam}}\le l_n(b_{\max})/(3.2\,b_{\max})$, since $(\mathrm{bay}-b)/(3.2\,b)$ is decreasing in $b$; in that case $\varsigma$ is literally frozen and the assumption reduces to monotonicity of the unprojected margin (Lemma~\ref{lem:clip}). Second, whether or not it holds, the returned design is feasible (Lemma~\ref{lem:safe}); only minimality along $b$ can be lost.

\begin{lemma}
\label{lem:bisect}
Under Assumption~\ref{ass:mono}: (i) the feasible section set $U(\theta_d)=\{b\in G_b:\ R(\Pi_{\mathrm{build}}(\theta_d,b))=1\}$ is an up-set of $G_b$; (ii) $b^\star(\theta_d)=\min U(\theta_d)$ exists if and only if $R(\theta_d,b_{\max})=1$, where $b_{\max}$ is the upper end of the grid after the cap of Eq.~\ref{eq:m-proj} (2.50~m when the cap is inactive); (iii) bisection locates $b^\star$ with at most $\lceil\log_2|G_b|\rceil=6$ evaluations, so one C2 call uses at most $2+\lceil\log_2|G_b|\rceil=8$ evaluations (upper-bound probe, bisection, final confirmation; Algorithm~\ref{alg:peloop}, lines~\ref{ln:probe}--\ref{ln:confirm}), and one outer round, which re-verifies the returned design once more, uses at most $\lceil\log_2 45\rceil+3=9$.
\end{lemma}
\begin{proof}
(i) If $b\in U$ and $b'\ge b$, then by Assumption~\ref{ass:mono} $V(\theta_d,b')\ge V(\theta_d,b)\ge 0$, hence $b'\in U$. (ii) By (i), $U\neq\varnothing$ if and only if $b_{\max}\in U$; $G_b$ is finite, so a minimum exists whenever $U$ is non-empty. (iii) Bisect the index interval $[0,h_0]$ with $b_{h_0}=b_{\max}$: if the midpoint is feasible, move the right endpoint to the midpoint, otherwise move the left endpoint to the midpoint plus one. By (i) the invariant ``the right endpoint is feasible and every index before the left endpoint is infeasible'' is preserved, the interval length at least halves at every step, and $h_0+1\le 45\le 2^{6}$ gives convergence to a single index within 6 steps. Adding the probe, the confirmation, and the final re-verification gives the stated bounds.
\end{proof}

\begin{lemma}
\label{lem:cost}
If the concrete volume $c(\theta_d,b)$ is strictly increasing in $b$, then under Assumption~\ref{ass:mono} $b^\star(\theta_d)=\operatorname*{arg\,min}\{c(\theta_d,b):\ b\in U(\theta_d)\}$, which is the second equality in Eq.~\ref{eq:m-bstar}.
\end{lemma}
\begin{proof}
$U$ is an up-set with minimum $b^\star$, and a strictly increasing function attains its minimum over $U$ at the minimum element.
\end{proof}

\begin{lemma}[Inactive clip]
\label{lem:clip}
If $\varsigma_{\mathrm{beam}}\le l_n(b_{\max})/(3.2\,b_{\max})$, then $\Pi_{\mathrm{build}}(\theta_d,b)=(\theta_d,b)$ for every $b\in G_b$ with $b\le b_{\max}$.
\end{lemma}
\begin{proof}
$l_n(b)/(3.2\,b)$ is decreasing in $b$, so its minimum over the capped grid is attained at $b_{\max}$; below that minimum the clip of Eq.~\ref{eq:m-proj} is the identity at every $b\le b_{\max}$, and the cap is already applied.
\end{proof}

\begin{lemma}[Feasibility without Assumption~\ref{ass:mono}]
\label{lem:safe}
Any design that C2 returns with status \textsc{pass} satisfies $R=1$, whether or not Assumption~\ref{ass:mono} holds; in the $\Pi_{\mathrm{build}}$ variant, it additionally satisfies $\min\boldsymbol{\rho}^{\mathrm{build}}\ge 0$.
\end{lemma}
\begin{proof}
C2 returns \textsc{pass} only after the confirmation evaluation (Algorithm~\ref{alg:peloop}, line~\ref{ln:confirm}) reports $R=1$ for the very design returned; in the $\Pi_{\mathrm{build}}$ variant, $\min\boldsymbol{\rho}^{\mathrm{build}}\ge0$ holds by construction of the projection.
\end{proof}

C2 proceeds in three steps (Algorithm~\ref{alg:peloop}, lines~\ref{ln:probe}--\ref{ln:confirm}). It first spends one evaluation on $(\theta_d,b_{\max})$; if $R=0$, then by Lemma~\ref{lem:bisect}(ii) $b^\star$ does not exist and the violated-constraint set $\mathcal{J}(\theta_d,b_{\max})$ is handed to C1. Otherwise, it bisects as in Lemma~\ref{lem:bisect}(iii), and finally spends one evaluation to confirm $b^\star$. The remaining budget hard-truncates the procedure; if truncation occurs and the returned section fails the confirmation, it is treated as ``$b^\star$ does not exist''. The section triple $\varsigma$ is frozen throughout the bisection (up to the beam-scale clip of $\Pi_{\mathrm{build}}$): enlarging $b$ monotonically stiffens the structure, whereas adding the beam section, concrete grade, or taper to the search would break monotonicity and multiply the evaluation count by the number of dimensions, so these three are left to the discrete decision in C1. If Assumption~\ref{ass:mono} fails on an individual instance (e.g., non-monotone behavior from higher-mode coupling), the final confirmation still guarantees that the returned $b^\star$ is feasible (Lemma~\ref{lem:safe}); only minimality is lost, and this risk is covered by the final verification and by the acceptance protocol of Section~\ref{sec:experiments}.
\subsection{C1 action set and the budgeted outer loop}
\label{app:c1}

C1 presents the language model with a geometry summary, the current $\theta_d$, and the violated-constraint set $\mathcal{J}$, and asks for one action from $\mathcal{A}$. The action set contains five action types on three levels of leverage: at the section level, adjusting the section triple $\varsigma$; at the topology level, adding a core, adding a shear wall, and adjusting the core plan geometry $\kappa$; at the system level, revising the declared lateral system in place and resetting the initial topology by Eq.~\ref{eq:m-intent}, i.e., an intent-level fallback. Actions only answer ``what to change''; out-of-range numerical parameters are clipped to their admissible ranges rather than rejected (core-to-floor area ratio $[0.01,0.30]$, plan aspect ratio $[0.2,5.0]$, plan offset $[-0.4,0.4]$ in each direction, beam-section scale $[0.5,2.0]$, concrete-grade scale $[1.0,1.30]$, taper $[0,0.40]$), so the output of C1 is always a valid point of $\Theta_d$. Adding a core or a wall increments the corresponding count; changing the lateral system resets the counts as in Eq.~\ref{eq:m-intent}; an unparsable action falls back to adding a core.

The outer loop runs as in Eq.~\ref{eq:m-outer}: round $k$ first executes C2 and passes if $b^\star$ is found; otherwise, if $k=K$ or the cumulative evaluation count satisfies $\mathrm{NFE}_k\ge B$, the episode terminates as \textsc{escalated}; otherwise, C1 is executed, and the next round begins. Fallback is thus layered: failures with the section below its upper bound are absorbed inside C2 by enlarging $b$; failures that persist at the upper bound are escalated to C1, where the language model chooses among section-, topology-, and system-level levers conditioned on $\mathcal{J}$ (for instance, torsion-governed failures favor system-level actions, but the model makes this choice and is not a hard-coded rule); \textsc{escalated} is declared only when the budget $B$ or the round limit $K$ is exhausted. This division confines the language model to the discrete semantic space where it has a comparative advantage, while the oracle still adjudicates the consequence of every decision.

\subsection{Material loop C3}
\label{app:c3}

C3 is enabled only in the \texttt{ultra} tier. Starting from the first passing design, which by Eq.~\ref{eq:m-bstar} already solves Eq.~\ref{eq:objective} along $b$, it accepts a candidate $\theta'$ iff
\begin{equation}
\theta^\star\leftarrow\theta'\quad\iff\quad R(\theta')=1\ \wedge\ c(\theta')<c(\theta^\star).
\label{eq:m-c3}
\end{equation}
To reduce $c$ along $\varsigma$, it enumerates five section-triple candidates (raising the concrete grade, introducing a vertical taper, enlarging the beam section, and their combinations), ordered by the material efficiency measured in a pilot study so that the cheapest lever is tried first, runs a coarse-to-fine column-side search around $b^\star$ for each candidate (steps of 12, 6, 3, and 1 grid cells, stopping at the first feasible reduction), re-checks every candidate with the oracle, and accepts strictly cheaper and only feasible candidates, following Eq.~\ref{eq:m-c3}. Candidates are evaluated independently from the pre-C3 pair $(\varsigma,b^\star)$ rather than cumulatively, and the cheapest feasible one is adopted (Algorithm~\ref{alg:peloop}). C3 has its own cap of 4 evaluations and is bounded by $B$; the candidate menu is deliberately small because every candidate consumes evaluations.

\subsection{Compute tiers}
\label{app:effort}

Structural design has no single correct compute budget: schematic design needs comparable options within minutes, whereas construction-document design needs a design that is right the first time. The architecture therefore does not fix the budget; it instantiates the same B0--C3 pipeline at four tiers that differ only in the budget $B$, the number of outer rounds $K$, and two lever switches: \texttt{light} $(B,K)=(5,1)$ for schematic comparison, \texttt{medium} $(10,2)$, \texttt{high} $(15,3)$ as the default and main-table setting, and \texttt{ultra} $(25,4)$ for construction documents. The C3 material loop runs only at \texttt{ultra}; the held-out L4 verification runs only at \texttt{high} and \texttt{ultra}, so we report the two lower tiers on PASS and evaluation count alone. The evaluation count includes only oracle calls; language-model and vision-model calls and the L4 final verification are not charged to $B$.

The tier values are operating points on one anytime family rather than derived quantities. The \texttt{high} tier is the main table setting: $B=15$ is the budget shared by every compared system, and $K=3$ allows one topology decision per lever level. The \texttt{light} tier sits below the cost of one complete round (9 evaluations, Lemma~\ref{lem:bisect}) and can therefore only screen feasibility; at $B=5$ the confirmation may be the last charged call, so the final measurement is skipped, and the confirmation margins are reported. \texttt{medium} affords one complete round plus a fallback; \texttt{ultra} adds one more round and the C3 budget. Section~\ref{sec:experiments} evaluates the family by a success rate that increases from \texttt{light} to the main-table \texttt{high} setting.
\subsection{Termination semantics}
\label{app:terminal}

An episode ends in one of four terminal states: \textsc{pass} (the returned design satisfies all in-loop constraints), \textsc{escalated} (honest termination after the levers or the budget are exhausted), \textsc{fail} (execution crash), and \textsc{unevaluated} (not assessable). \textsc{escalated} and \textsc{fail} are strictly separated: the former is the legitimate, scored terminal state of Section~\ref{sec:task}, in which the agent declares that no feasible design was found within the given action set and budget, rather than a feigned convergence; the latter is a system failure. \textsc{unevaluated} is further attributed to the responsible side: non-delivery, timeout, or crash on the side of the system under test is counted as infeasible and enters the denominator, whereas cases in which the harness cannot synthesize the geometry or has no sourced constraint are removed from the denominator. This attribution rule prevents the pass rates in Section~\ref{sec:experiments} from being systematically biased against any system by harness defects.

\subsection{Final verification L4}
\label{app:l4}

Dynamic response is the real loading condition, but response-history analysis under recorded ground motions is one to two orders of magnitude more expensive than elastic static analysis, and placing it inside the iteration loop would raise the cost per instance by an order of magnitude; seismic response-history analysis is therefore used only as a final check. After the equivalent-static loop converges, the returned design $\theta^\star$ undergoes one seismic response-history analysis under recorded ground motions and one full wind time-history analysis, yielding the held-out margin vector of Eq.~\ref{eq:m-l4}. Its components are the inter-story drift ratio under response history (same limit as the in-loop constraint) and the wind-induced peak acceleration (limited by occupancy: $0.15$~m/s$^2$ for residential and apartment buildings and $0.25$~m/s$^2$ for office and hotel buildings, following Clause~3.7.6 of JGJ~3~\citep{jgj3}, applied above its 150~m applicability height). Whether seismic response-history analysis is mandatory is prescribed by height, seismic intensity, and site class (Clause~5.1.2 of GB~50011~\citep{gb50011}: above 100~m for intensity~7 and for intensity~8 on site classes I--II, above 80~m for intensity~8 on site classes III--IV, above 60~m for intensity~9); our held-out verification runs the seismic response history on every returned design without applying this trigger, so that the held-out set is not systematically biased toward instances that are unlikely to fail. Because the recorded-ground-motion condition never enters $G_x$, L4 is a held-out test on a loading condition unused inside the loop; it shares the physical model with the loop and therefore cannot detect systematic errors in the model form itself (Section~\ref{sec:limitations}). The third certificate component $\boldsymbol{\rho}^{\mathrm{build}}$ is the detailing margin vector in Appendix~\ref{app:build}; for the projected agent, it is non-negative by construction, and for non-projected systems, it records the detailing gap.

\paragraph{Held-out records and scaling.}
The seismic re-check draws one record per design from the FEMA P-695 far-field subset of the record library of Appendix~\ref{app:bench}: candidates are ranked by spectral-shape misfit to the site design spectrum over $[0.2\,T_1,\,1.5\,T_1]$, the best match is taken, and it is amplitude-scaled by the geometric-mean log ratio of code to record spectral acceleration over that band, the convention of GB~50011 Clause~5.1.2, ASCE~7 Chapter~16, and EN~1998-1 Clause~3.2.3.1~\citep{gb50011,asce7,en1998}. The analysis is an elastic modal-superposition response history in one horizontal direction with 5\% Rayleigh damping anchored at $T_1$ and $T_3$, restricted to the 5--95\% Arias-intensity window; the peak inter-story drift ratio under that record is checked against the instance's drift limit. The re-check runs on every returned design, whether or not Clause~5.1.2 of GB~50011 would require a time-history analysis at that height. The wind re-check draws three of the sixteen ten-minute series of the wind library of Appendix~\ref{app:bench}, their fluctuation amplitude rescaled to the turbulence intensity of the site's terrain category, field records taken first, and applied as quasi-steady buffeting loads $F_i(t)=\bar F_i\,[v(t)/\bar v]^2$ at the code basic wind pressure. The same modal machinery with 2\% damping gives the drift envelope, checked against the static limit unless the first natural frequency exceeds 1~Hz, in which case the code static wind load governs drift, and the top-floor acceleration, converted from the fifty-year to the ten-year return pressure of GB~50009 Table~E.5~\citep{gb50009} by linear rescaling and checked against the Clause~3.7.6 limits of JGJ~3~\citep{jgj3}; for code families whose ten-year pressure is not tabulated, the acceleration check is reported as unevaluated rather than approximated.

\subsection{Architectural placement}
\label{app:graph}

From a graph-engineering perspective, the pipeline is a task graph whose nodes are deterministic physics tools and language-model calls (perception, intent parsing, topology decision, parameter search, material optimization, final verification) and whose edges represent data dependencies. The graph runs as an iterative loop rather than a one-shot orchestration. In each round, the agent decides the next action (repair or escalate) based on physics feedback rather than being forced along the graph structure. In the layered stack of graph engineering~\citep{feng2026graph,he2026awesomegraphengineering}, the system sits at the ``heavy harness, single loop'' position: the acceptance authority is built into the deterministic physics tools of the harness layer rather than delegated to a supervising model of the same origin, which is precisely the structural remedy for the verification-credibility deficit discussed in Appendix~\ref{sec:related}.

\section{Additional experimental results}
\label{app:exp}

\paragraph{Protocol.}
The 260-instance main split holds 130 topology-decision instances, 110 section-only controls, and 20 capability-boundary instances; the last are excluded from every comparative statistic and listed in Table~\ref{tab:strata} only to document that no arm solves them. Within an LLM block, every arm faces the same instances on every seed, the same margin oracle (Eq.~\ref{eq:margin}), and the same budget; each rate divides by its own denominator (below). The search-sufficient reference set (GB-59: 59 instances $\times$ 3 seeds) carries the comparisons against the LLM-free reference lines and is never pooled or compared with the main split (caveat~(ii) of Appendix~\ref{app:threats}).

\paragraph{Representativeness of the main split.}
The 130 scored instances are a stratified 37.7\% sample of the 345 topology-decision instances in the pool, the only stratum that separates systems: section-only instances are passed by 93.6--97.9\% of answers under every system through column resizing, and boundary instances by none (Table~\ref{tab:strata}). Sampling is deliberately non-uniform by height, drawing 4.8\% of the 1,239 instances below 24~m but 87.0\% of the 23 above 250~m, which is what makes the twenty boundary instances available; the split thus reflects the discriminative stratum, not the pool's height distribution. At $n{=}390$ answers per cell, the Agresti--Coull standard error \citep{agresti1998approximate} is about 2--3~pp, several times smaller than the reported effects, so enlarging the split would narrow already narrow intervals. We release the full corpus, including the 1,035 boundary instances; Tables~\ref{tab:diffmap} and~\ref{tab:strata} give the map, the quotas, and the unscored strata.

\begin{table}[!t]
\centering\small
\caption{Difficulty map of the screened corpus and the main-split quotas. Every one of the 3,216 screened instances was labeled by an exhaustive LLM-free search that tries section resizing, then up to four cores, then up to four walls, with at most 45 evaluations per stage: section-only instances are solved by resizing alone, topology-decision instances need a core or wall, and boundary instances are not solved within the action space. The last three columns give how many of each the 260-instance main split draws per height band.}
\label{tab:diffmap}
\setlength{\tabcolsep}{4pt}
\begin{tabular}{lrrrrrrrr}
\toprule
& \multicolumn{4}{c}{Screened corpus} & & \multicolumn{3}{c}{Main split} \\
\cmidrule{2-5}\cmidrule{7-9}
Height band & $n$ & Section-only & Topology-decision & Boundary & & T-dec. & Controls & Boundary \\
\midrule
$<$24~m & 1,239 & 828 & 49 & 362 & & 30 & 30 & 0 \\
24--60~m & 968 & 580 & 193 & 195 & & 65 & 35 & 0 \\
60--120~m & 615 & 357 & 98 & 160 & & 30 & 30 & 0 \\
120--250~m & 371 & 71 & 5 & 295 & & 5 & 15 & 0 \\
$>$250~m & 23 & 0 & 0 & 23 & & 0 & 0 & 20 \\
\midrule
All & 3,216 & 1,836 & 345 & 1,035 & & 130 & 110 & 20 \\
\bottomrule
\end{tabular}
\end{table}

\begin{table}[!htb]
\centering\small
\caption{Static PASS\% on the three strata of the main split under DeepSeek-V4-Flash (topology-decision column as in Table~\ref{tab:m1}; the other two strata are tabulated from seed~0 for the baselines and from all seeds for PE-Loop). Section-only controls are passed by 93.6--97.9\% of answers under every system and boundary instances by none, so only the topology-decision stratum separates systems.}
\label{tab:strata}
\begin{tabular}{lrrr}
\toprule
Arm & Topology-decision (130) & Section-only controls (110) & Boundary (20) \\
\midrule
Codex & 82.05$\pm$1.95 & 93.64 & 0.00 \\
Openclaw & 73.85$\pm$2.22 & 94.55 & 0.00 \\
DeepSeek Harness & 82.82$\pm$1.91 & 95.45 & 0.00 \\
Claude Code & 83.85$\pm$1.87 & 97.27 & 0.00 \\
\textbf{PE-Loop (ours)} & \textbf{91.28}$\pm$1.45 & \textbf{97.88} & 0.00 \\
\bottomrule
\end{tabular}
\end{table}

\paragraph{Column definitions for Table~\ref{tab:m1}.}
Columns: PASS = static pass rate on the topology-decision instances ($\uparrow$); ESC = escalation rate (diagnostic, no preferred direction); NFE@p = NFE@pass, mean evaluations of passing answers as a fraction of the budget ($\downarrow$); Exh.\ = NFE exhausted, share of answers that run the budget to zero ($\downarrow$); Wall = median wall-clock minutes, LLM gateway queueing included ($\downarrow$); T/O = timeout rate ($\downarrow$); pass$^3$ = share passing on all three seeds, the pass$^k$ reliability metric of \citet{yao2025taubench} at $k{=}3$ ($\uparrow$); Valid = validity rate ($\uparrow$); Seis.\ / Wind = held-out seismic / wind time-history re-check ($\uparrow$); Build = constructability gate ($\uparrow$); Verif.\ = Verified PASS, the four-way conjunction of PASS, Seis., Wind, and Build ($\uparrow$); P/h = passes per hour ($\uparrow$); Rank = mean, over the twelve ranked columns, of the arm's rank within its block (1 = best; ties are broken by table order; ESC is excluded) ($\downarrow$). All columns are percentages except Wall (minutes), P/h (count per hour), and Rank.

\paragraph{Systems and run conditions.}
Five arms are compared on every backbone: PE-Loop; Codex, OpenAI's \texttt{codex exec} command-line agent; Openclaw, an open-source agent harness; DeepSeek Harness, the model vendor's own scaffold (version 0.1.0-rc.7); and Claude Code (version 2.1.237), driven through an Anthropic-protocol bridge to the same backbone with two known asymmetries that disadvantage it and that we did not remove: its system prompt names a different model than the one serving it, and it truncates context to 200k tokens although the backbone offers more. The four baselines receive the identical brief, shell access, and the same deterministic margin oracle as a callable tool under the same harness-enforced budget (15 at effort \texttt{high}); they decide when to stop and what to submit, and the returned design is read from a file they write. Every episode is capped at 1,800~s, and timeouts count as failures; all arms reach one API endpoint through a shared model-level rate-limiting gateway; sampling parameters are provider defaults, since the command-line agents expose no temperature or seed control, and seeds 0--2 are replicate indices. Three cells deserve a note: under GLM~5.2~Fast, 62 completed DeepSeek Harness episodes accumulated at least 300~s of gateway queueing, inflating that cell's Wall, and seeds 1--2 of Openclaw and Claude Code ran through a relay path rather than the central gateway, with different completion rates; under Qwen3.8-Flash, the PASS of Claude Code differs by 14.62~pp between seeds 0 and 2 (2.4 SE), below the 3-SE threshold of our data gate. The campaign comprised about 15,600 agent episodes between 3 and 10 September 2026; the held-out analyses took about 430 CPU core-hours. The 372-task quarantine set and the 1,197 paired variants are reserved for contamination audits and were not used here.

\paragraph{Reading notes for Table~\ref{tab:m1}.}
Every arm answers the full set on all three seeds under every backbone ($n{=}390$ topology-decision, $n{=}720$ all, $n{=}130$ for pass$^3$), and each rate divides by the answers the arm was run on. Under GLM 5.2 Fast, Openclaw, DeepSeek Harness, and Claude Code lost 5, 14, and 12 topology-decision answers to the relay path and gateway queueing; these non-deliveries count as failures in every denominator, which disfavors these baselines and is disclosed rather than imputed. PE-Loop's exhaustion stays within 0.00--6.41\% because it escalates after the fixed probing phase rather than running the budget down; the baselines have no budget-aware escalation action, so an answer ending without a feasible submission before the budget runs out counts as ESC, and most run to the budget. Seis., Wind, and Build share the PASS denominator, so a design that is not statically PASS fails them too; Verified PASS is their conjunction with PASS. Because every arm answers the same instances on every seed, Table~\ref{tab:m1} compares arms on identical instance sets within a block; we report unpaired rates with their standard errors.

\paragraph{Why the constructability gate binds for PE-Loop.}
The cause is lever order (Appendix~\ref{app:c1}): a failing candidate is first repaired inside C2 by enlarging $b$, and only a failure at the top of the column grid escalates to a topology change in C1, so the delivered column is as thick as the section lever alone required. Since $h_b=0.8\,b\,\varsigma_{\mathrm{beam}}$ grows with $b$ while $l_n$ shrinks (Appendix~\ref{app:build}), thick-column designs violate $l_n/h_b\ge 4$ first, precisely the failure mode $\Pi_{\mathrm{build}}$ removes.

\paragraph{Where the certified margin moves.}
Verified PASS is a conjunction, so marginal gate rates do not add; we compare survival conditional on static compliance. Under DeepSeek-V4-Flash, 66.6\% of our compliant designs pass the held-out seismic check and 76.7\% the wind check (70.0\% and 76.1\% for Claude Code), while 50.8\% clear the detailing clause (62.1\%). The certified lead of 4.61~pp is therefore what survives of a 7.43~pp raw lead after a constructability deficit that $\Pi_{\mathrm{build}}$ removes (Table~\ref{tab:pibuild}). Figure~\ref{fig:gates-all} shows the same construction for every block and system.

\paragraph{Detailing gap by column size.}
The gap grows with column size and is largest for PE-Loop without $\Pi_{\mathrm{build}}$, whose outer loop pulls the section lever first (median column side 1.275~m against 0.950--1.200~m for the baselines). Of its 126 violations, 123 (97.6\%) are repairable by clipping the beam-depth scale alone; the other 3 (2.4\%) are topology-level cases with $b>\mathrm{bay}/2.6$. Conditioned on static compliance (Build\% over PASS\% in Table~\ref{tab:m1}, DeepSeek-V4-Flash), the gate is harshest on us: 50.84\% of our compliant designs clear the clause, against 61.88--65.33\% for the four baselines; with $\Pi_{\mathrm{build}}$ in the loop the rate is 100\% by construction (Table~\ref{tab:pibuild}).

\begin{table}[!t]
\centering\small
\caption{Grounding ablation under DeepSeek-V4-Flash: same loop, same budget, same 60 pre-registered topology-decision instances $\times$ 3 seeds; only the gate implementation varies. The ablation is a separate run from Table~\ref{tab:m1}; on the same 180 answers, the main-table run of PE-Loop passes 81.67\% (147/180), within sampling variance of the oracle arm here ($z\approx0.92$). These instances are harder than the full stratum: they hold 33 of the 34 failures of that run. The judge receives the brief, the code limits, and the design parameters, never a simulator quantity. Verified PASS is the harness's deterministic physics check on the submitted design, identical for every arm; for the oracle arm it coincides with the in-loop gate by construction and serves as the lossless reference. Whether the in-loop constraint set is itself sufficient is a separate question, answered by the held-out columns of Table~\ref{tab:m1}; the judge arms were not re-checked under held-out loading.}
\label{tab:m3}
\setlength{\tabcolsep}{3pt}
\resizebox{\textwidth}{!}{%
\begin{tabular}{lrrrrrr}
\toprule
Gate implementation & Self-rep.\ PASS\% & Verified PASS\% $\uparrow$ & False-pos.\% $\downarrow$ & Wasted NFE\% $\downarrow$ & Jaccard $\uparrow$ & Wall s \\
\midrule
Margin oracle (ours) & 77.78\% & \textbf{77.78\%} & \textbf{0.00\%} & \textbf{0.00\%} & \textbf{1.00} & \textbf{110.8} \\
LLM-judge (same LLM) & 49.44\% & 20.56\% & 58.43\% & 44.41\% & 0.31 & 238.4 \\
LLM-judge (DeepSeek-V4-Pro) & 40.00\% & 17.22\% & 56.94\% & 45.19\% & 0.29 & 342.8 \\
\bottomrule
\end{tabular}
}
\end{table}

\begin{table}[!t]
\centering\small
\caption{Constructability projection ablation (main split, topology-decision instances, DeepSeek-V4-Flash). Raw PASS\% drops by 29.74~pp; certified success rises by 9.74~pp. With $\Pi_{\mathrm{build}}$, every statically compliant design satisfies the detailing clause by construction, so Build\% equals PASS\%, and the detailing gap is zero.}
\label{tab:pibuild}
\setlength{\tabcolsep}{3pt}
\resizebox{\textwidth}{!}{%
\begin{tabular}{lrrrrrrr}
\toprule
Arm & PASS\% & Concrete m$^3$ & Build\% & Detailing gap (pp) & THA seis.\% & THA wind\% & Verified PASS\% \\
\midrule
PE-Loop & 91.28\% & 5,610.1 & 46.41\% & $+44.87$ & 60.77\% & 70.00\% & 35.64\% \\
PE-Loop $+\ \Pi_{\mathrm{build}}$ & 61.54\% & 5,032.0 & 61.54\% & $+0.00$ & 51.79\% & 53.59\% & \textbf{45.38\%} \\
\bottomrule
\end{tabular}
}
\end{table}

\paragraph{Scope of the projection ablation.}
Three limits apply to Table~\ref{tab:pibuild}: this is the main split, not the 59-instance reference set; one LLM only (DeepSeek-V4-Flash); and it is scored on the 130 topology-decision instances, while the unprojected arm answers the full 260, so it is not a full-coverage control.

\paragraph{Caliber statements.}
On the search-sufficient reference set, PE-Loop returns 175 answers (174 feasible) and Openclaw, DeepSeek Harness, and Claude Code 176 each, all on 59 instances; Codex returns 149 answers on 51 instances (131 feasible). Cross-system claims use the 51-instance intersection, and paired tests against search use each arm's own coverage (Table~\ref{tab:gb59cov}). Checks whose code limit could not be sourced, or whose quantity the agent never reported, are \textsc{unevaluated} and leave the denominator rather than counting as failures. Two criteria carry no discriminative signal and are not discussed in Section~\ref{sec:experiments}: intent-consistency saturates at 98.66--100.00\% (range 1.34~pp), and the beam-width rule binds only for $\varsigma_{\mathrm{beam}}>1.875$ with taper $>1-0.32\,\varsigma_{\mathrm{beam}}$, which no returned design reaches.

\begin{table}[!htb]
\centering\small
\caption{Coverage on the search-sufficient reference set (59 instances $\times$ 3 seeds; 177 answers maximum). The last column is the count and within-arm share of feasible answers that violate the clear-span-to-depth clause (Fig.~\ref{fig:build}).}
\label{tab:gb59cov}
\begin{tabular}{lrrrr}
\toprule
Arm & Instances & Answers & Feasible & Violating $l_n/h_b\ge4$ \\
\midrule
Codex & 51 & 149 & 131 & 56 (42.75\%) \\
Openclaw & 59 & 176 & 170 & 85 (50.00\%) \\
DeepSeek Harness & 59 & 176 & 169 & 86 (50.89\%) \\
Claude Code & 59 & 176 & 171 & 76 (44.44\%) \\
PE-Loop w/o $\Pi_{\mathrm{build}}$ & 59 & 175 & 174 & 126 (72.41\%) \\
\bottomrule
\end{tabular}
\end{table}

\subsection{The search-sufficient reference set: 59 instances}

\paragraph{Paired significance and failure modes.}
Against LLM-free exhaustive search (exact McNemar test on feasibility with Holm correction \citep{holm1979simple}, $N{=}59$), only Codex differs significantly: 18 instances that search solves and Codex does not, and none the other way ($p_{\mathrm{Holm}}=1.14\times10^{-4}$). Openclaw, PE-Loop, DeepSeek Harness, and Claude Code have 2, 0, 0, and 1 discordant instances, respectively ($p_{\mathrm{Holm}}=1$); $p=1$ with no discordant pair is agreement instance by instance, not evidence of equivalence, because the paired test's denominator is the discordant set. Per criterion, with \textsc{unevaluated} checks excluded from the denominator, the failures of Codex concentrate on story drift (12.08\%, 18/149), stability (10.74\%, 16/149), and the base-shear floor (8.05\%, 12/149), with the story-stiffness ratio at 1.34\%; every other system fails only the story-stiffness ratio (1.69--1.74\%, and never Openclaw) and, at most once each, the period ratio; torsion and wind acceleration are never failed by any system.

\subsection{Effort tiers and the anytime contract}

All numbers here are seed-0 values on the 130 topology-decision instances (Table~\ref{tab:m1} pools three seeds, hence its different Wall of 1.6~min). Scoring the four tiers of Appendix~\ref{app:effort} on these identical instances removes the set-composition confound; Fig.~\ref{fig:results}d plots PASS against $B$ with NFE@pass under each point. Wall-clock is not monotone across tiers (175 / 183 / 104 / 258~s, seed-0 median) because it includes LLM gateway queueing, and neither is success: \texttt{ultra} ($B{=}25$) returns 81.54\% against 91.54\% (119 of 130) at \texttt{high} ($B{=}15$), because the extra rounds let C1 revise the topology of instances the section lever had already settled, not always for the better.

\subsection{Denominator disclosure}

Every rate divides by the answers each arm was run on: $720$ per arm and backbone ($260-20$ instances $\times$ 3 seeds) and $180$ per arm of the grounding ablation. The 31 GLM-block non-deliveries count as failures and stay in the denominator; otherwise the effective denominator equals the answer count in every row. Agent-side \textsc{unevaluated} answers (at most eleven, for Codex; four for Openclaw, one each for DeepSeek Harness and Claude Code, and three for the stronger judge) also stay in the denominator, as do all timeouts. Rate-limit events were retried rather than dropped, so they do not reduce the denominator, but they are uneven: $0$ for PE-Loop, $34$ for Openclaw, $90$ for Codex, $134$ for Claude Code, and $297$ for DeepSeek Harness; three of the four baselines thus spent a substantial share of their runs waiting on a gateway, which the Wall column of Table~\ref{tab:m1} includes.

\subsection{Threats to validity of these numbers}
\label{app:threats}

Three caveats apply; none can be resolved with more seeds. (i)~For most instances, the plan footprint behind the column grid is back-derived from floor area at a fixed aspect ratio of 1.3 (\texttt{geom\_source} = \texttt{from\_area}), biasing the bay width and, with it, the detailing margin $\rho$ in Fig.~\ref{fig:build} for multi-block massings. (ii)~The main split and the reference set differ in instances, seeds, and run directories, so their numbers are never comparable. (iii)~In every block, the cross-arm PASS\% range exceeds the pre-registered 15~pp separation (narrowest: 17.43~pp, DeepSeek-V4-Flash and GLM 5.2 Fast), but our strongest evidence is the four-LLM grid, held-out columns, and paired tests.

\end{document}